\documentclass{amsart}
\usepackage{amsthm}
\usepackage{amsmath}
\usepackage{amssymb}

\usepackage{array}
\usepackage{multirow}
\usepackage{geometry}
\allowdisplaybreaks

\DeclareMathOperator{\Tr}{Tr}

\usepackage{color}

\newtheorem{thm}{Theorem}[section]{\bfseries}{\itshape}
\newtheorem{cor}[thm]{Corollary}{\bfseries}{\itshape}
\newtheorem{prop}[thm]{Proposition}{\bf}{\it}
\newtheorem{lem}[thm]{Lemma}{\bf}{\it}
{\bf}{\it}

{\bfseries}{\rmfamily}
\newtheorem{rem}[thm]{Remark}{\bfseries}{\rmfamily}

\makeatother

\title[mean-field evolution with singular three-body interactions]{Rate of convergence towards mean-field evolution for weakly interacting
	bosons with singular three-body interactions}
\author{Jinyeop Lee}
\address{School of Mathematics, Korea Institute	for Advanced Study,	Seoul 02455, Republic of Korea}
\email{jinyeoplee@kias.re.kr}

\address{Department of Mathematics, LMU Munich, Theresienstrasse 39, 80333 Munich, Germany}
\email{lee@math.lmu.de}

\begin{document}
\maketitle              
\begin{abstract}
	In this paper, we investigate the dynamics of a system of $N$ weakly
	interacting bosons with singular three-body interactions in three
	dimensions. By assuming factorized initial data $\Psi_{N,0}=\varphi_{0}^{\otimes N}$
	and triple collisions, we prove that in the many-particle limit, its
	mean-field approximation converges to quintic Hartree dynamics. Moreover,
	we prove that the rate of convergence towards the mean-field quintic
	Hartree evolution is of $O(N^{-(1+4a)/(3+2a)})$ for $\varphi_{0}\in H^{(3/2)+a}(\mathbb{R}^{3})$
	where $0\leq a<1/2$ and $O(N^{-1})$ for $a>1$. Our proof is based
	on and extends the Fock space approach.
\end{abstract}


\section{Introduction\label{sec:introduction}}


We are consider in the time-evolution of a condensate of $N$-particles
under the presence of a three-body interaction in the mean-field regime. 
We prove that, in the many-particle limit, the solution of the many-particle Schr\"odinger
equation with a three-body interaction can be approximated by its
mean-field limit (quintic Hartree equation) in the trace norm sense.

We consider a system of $N$ bosonic particles interacting with three-body
interaction having factorized initial data, i.e.,

\begin{align}
\mathrm{i}\partial_{t}\Psi_{N,t} & =\left(\sum_{j=1}^{N}(-\Delta_{x_{j}})+\frac{1}{N^{2}}\sum_{1\leq i<j<k\leq N}V(x_{i}-x_{j},x_{i}-x_{k})\right)\Psi_{N,t}\label{eq:N_body_Hamiltonian}\\
\left.\Psi_{N}\right|_{t=0} & =\prod_{j=1}^{N}\varphi_{0}(x_{j}).\nonumber 
\end{align}
Note that $\Psi_{N,t}\in L_{\text{sym}}^{2}(\mathbb{R}^{3N})$. The
mean-field evolution of the above system is given by the following
quintic Hartree system
\begin{align*}
\mathrm{i}\partial_{t}\varphi_{t} & =-\Delta\varphi_{t}+\frac{1}{2}\left(\int\mathrm{d}y\mathrm{d}z\,V(x-y,x-z)|\varphi_{t}(y)|^{2}|\varphi_{t}(z)|^{2}\right)\varphi_{t}\\
\left.\varphi_{t}\right|_{t=0} & =\varphi_{0}.
\end{align*}
To generalize the Coulomb interactions, it is natural to assume the
three-body interaction potential $V(x-y,x-z)$ is related to the distance
between $x$, $y$, and $z$. Moreover, it should be symmetric up
to the variables $x$, $y$, and $z$. Hence, we will assume
\begin{equation}
V(x-y,x-z)=\lambda(v(x-y)v(x-z)+v(y-z)v(y-x)+v(z-x)v(z-y))
\end{equation}
for some constant $\lambda>0$. We assume $v$ to be the
Coulomb potential, i.e., $v(x)=1/|x|$. Note that $v$ satisfies the
following operator inequality
\[
v(x)\leq C(1-\Delta)
\]
by Hardy inequality.

To understand our system rigorously at time $t\geq0$, we proceed
as follows. First, we consider the density matrix $\gamma_{N,t}=|\Psi_{N,t}\rangle\langle\Psi_{N,t}|$
associated with $\Psi_{N,t}$, which can be understood as the orthogonal
projection onto $\Psi_{N,t}$. More precisely, the kernel of $\gamma_{N,t}$
is given by 
\[
\gamma_{N,t}(\mathbf{x}_{N};\mathbf{x}_{N}'):=\Psi_{N,t}(\mathbf{x}_{N})\overline{\Psi_{N,t}(\mathbf{x}_{N})}
\]
where we denote $\mathbf{x}_{m}\in\mathbb{R}^{3m}$ for any $m\in\mathbb{N}$.
The $k$-particle marginal density is then defined through its kernel
\begin{equation}
\gamma_{N,t}^{\left(k\right)}(\mathbf{x}_{k};\mathbf{x}'_{k})=\int\mathrm{d}\mathbf{x}_{N-k}\gamma_{N,t}(\mathbf{x}_{k},\mathbf{x}_{N-k};\mathbf{x}'_{k},\mathbf{x}_{N-k}).\label{eq:Kernel_of_Marginal_Density}
\end{equation}
We now focus on the trace-norm distance between the one-particle marginal
density $\gamma_{N,t}^{(1)}$ and the projection operator $|\varphi_{t}\rangle\langle\varphi_{t}|$.
In particular, we will prove that 
\begin{equation}
\Tr\left|\gamma_{N,t}^{(1)}-|\varphi_{t}\rangle\langle\varphi_{t}|\right|\leq\frac{Ce^{Kt}}{N}\label{eq:trace_norm_bound}
\end{equation}
and find $C$ according to the conditions on $V$ and $\|\varphi_{0}\|$.
The norm $\|\cdot\|$ of $\varphi_{0}$ will be discussed later. 
	
\begin{thm}
	\label{thm:main}Let the three-body interaction potential $V(x-y,x-z)$
	such that
	\[
	V(x-y,x-z)=\lambda(|x-y|^{-1}|x-z|^{-1}+|y-z|^{-1}|y-x|^{-1}+|z-x|^{-1}|z-y|^{-1})
	\]
	where $\lambda>0$. For $\varphi_{0}\in H^{(3/2)+a}(\mathbb{R}^{3})$
	such that $\|\varphi_{0}\|_{L^{2}(\mathbb{R}^{3})}=1$, and $\varphi_{t}$
	be the solution of the quintic Hartree equation
	\begin{equation}
	\mathrm{i}\partial_{t}\varphi_{t}=-\Delta\varphi_{t}+\frac{1}{2}\left(\int\mathrm{d}y\mathrm{d}z\,V(x-y,x-z)|\varphi_{t}(y)|^{2}|\varphi_{t}(z)|^{2}\right)\varphi_{t}\label{eq:quinticHartree}
	\end{equation}
	with initial data $\varphi_{t=0}=\varphi_{0}$. Let $\psi_{N,t}=e^{-\mathrm{i}H_{N}t}\varphi_{0}^{\otimes N}$
	and $\gamma_{N,t}^{(1)}$ be the one-particle reduced density associated
	with $\psi_{N,t}$ as defined in \eqref{eq:Kernel_gamma}. Then there
	exist constants $C$ and $K$, depending only on $\|\varphi_{0}\|_{H^{(3/2)+a}}$,
	and $\lambda$ such that
	\begin{enumerate}
	\item if $0\leq a<1/2$, then
	\[
	\Tr\left|\gamma_{N,t}^{(1)}-|\varphi_{t}\rangle\langle\varphi_{t}|\right|\leq\frac{Ce^{Kt}}{N^{(1+4a)/(3+2a)}},
	\]
	\item if $a>1$, then
	\[
	\Tr\left|\gamma_{N,t}^{(1)}-|\varphi_{t}\rangle\langle\varphi_{t}|\right|\leq\frac{Ce^{Kt}}{N}.
	\]
\end{enumerate}
\end{thm}

\vspace{2em}

%
%

\begin{rem}
	One can obtain the same result by using the other three-body interaction
	potentials such that
	\[
	V(x-y,x-z)=V_{\infty}(x-y,x-z)V_{C}(x-y,x-z)
	\]
	where $V_{\infty}\in L^{\infty}(\mathbb{R}^{9})$ and $V_{C}(x-y,x-z)=\lambda(|x-y|^{-1}|x-z|^{-1}+|y-z|^{-1}|y-x|^{-1}+|z-x|^{-1}|z-y|^{-1})$
	or
	
	\[
	V(x-y,x-z)=V_{3}(x-y,x-z)+V_{\infty}(x-y,x-z)
	\]
	where $V_{3}(x-y,x-z)=v_{3}(x-y)v_{3}(x-z)+v_{3}(y-z)v_{3}(y-z)+v_{3}(z-x)v_{3}(z-y)$
	with $v_{3}\in L^{3}(\mathbb{R}^{3})$ and $V_{\infty}\in L^{\infty}(\mathbb{R}^{9})$.
\end{rem}


\begin{rem}
	In \cite{XChen2012}, the author assume that
	\[
	V(x-y,x-z)=v(x-y)v(x-z)+v(y-x)v(y-z)+v(z-x)v(z-y)
	\]
	where, for $\varepsilon\in(0,1/2)$,
	\[
	v(x):=\frac{\chi(|x|)}{|x|^{1-\varepsilon}}\quad\text{or}\quad G_{2+\varepsilon}(x),
	\]
	where $\chi\in C_{0}^{\infty}(\mathbb{R}^{+}\cup\{0\})$ is nonnegative
	decreasing function and $G_{\alpha}$, which is the kernel of the
	Bessel potential. 
	Succeeding that work, through out
	this paper we consider the case where $\varepsilon=0$ without assuming
	the fast decay. Moreover, we do not assume the initial data to have finite variance.
\end{rem}

\vspace{2em}

Before we describe the main ideas used in the proof of Theorem \ref{thm:main},
we give a brief summary of the related known results. The derivation
of the mean-field limit of a dilute Bose gas has been actively studied.
First, profound results \cite{Ginibre1979_1,Ginibre1979_2,Hepp1974,Spohn1980}
give us, using the BBGKY hierarchy, that the convergence, 
\[
\Tr\left|\gamma_{N,t}^{(1)}-|\varphi_{t}\rangle\langle\varphi_{t}|\right|\to0\quad\text{as}\quad N\to\infty.
\]
Form Erd\H os and Yau \cite{ErdosYau2001}, the convergence is also
proven for a singular potential (including the Coulomb case). Rodnianski
and Schlein in \cite{Rodnianski2009}, developed a coherent state
approach to obtain a bound for the trace norm difference 
\[
\Tr\left|\gamma_{N,t}^{(1)}-|\varphi_{t}\rangle\langle\varphi_{t}|\right|\leq O(N^{-1/2})
\]
for singular potentials (including Coulomb potential). The proof is
based on the works of Fröhlich \cite{FrohlichGraffiScharz2007,FrohlichKnowlesPizzo2007,FrohlichKowlesSchartz2009}.
The rate of convergence $O(N^{-1})$ is known to be optimal \cite{Grillakis10,Grillakis11,Grillakis13,Grillakis17}.
The optimal rate of convergence $O(N^{-1})$ was obtained in \cite{ChenLeeSchlein2011}
for the Coulomb case. A similar approach has been applied to many-body
semi-relativistic Schrödinger equations with gravitational interaction
\cite{Lee2013}. Methods in Fock space have been studied for the dynamical
properties of a BEC \cite{Benedikter2015,Boccato2018CompleteBEC,Boccato2017,Brenneck2019,ChenLee2011,ChenLeeLee2018,ChenLeeSchlein2011,LewinNamRougerie2016,Rodnianski2009}.
The same rate of convergence an be obtained by counting the number
of particles in the condensate state \cite{KnowlesPickl2010,MitrouskasPetratPickl2019,Pickl2011}.

In \cite{ChenPavlovic2011}, a system with three-body interactions
is considered. Chen and Pavlovi\'c proved that, in the Gross-Pitaevskii
regime, the limiting dynamics is governed by the quintic nonlinear
Schr\"odinger equation using the BBGKY hierarchy method in dimension
$d=1,2$, i.e., 
\[
\Tr\left|\gamma_{N,t}^{(1)}-|\varphi_{t}\rangle\langle\varphi_{t}|\right|\to0\quad\text{as}\quad N\to\infty.
\]
For $d=3$, Nam and Salzmann derived quintic nonlinear Sch\"odinger equation from Gross–Pitaevskii scaling, i.e., 
\[
V_N(x-y,x-z) = N^{6\beta}V(N^\beta (x-y),N^\beta (x-z))
\]
with $0<\beta<1/6$ with the initial data in higher Sobolev space, to be more specific, $\varphi_0\in H^4(\mathbb{R}^3)$.
Moreover, Chen and Holmer provide the derivation of the energy critical nonlinear Sch\"odinger equation in \cite{ChenHolmer2019}.

In \cite{XChen2012}, Chen provides a rigorous proof for Hartree dynamics
under the presence of triple (repulsive) collisions with singular
interaction potential in the mean-field limit. He provided the rate in the Fock space norm sense for GMM type approximation. The provided rate is $O(N^{-1/2})$.

In the two-body interaction case, as in \cite{ChenLeeSchlein2011,Lee2013,Rodnianski2009},
we first embed the initial state to the Fock space replacing it by
the coherent state. For the evolution of the coherent state, we need
to control the fluctuation $\mathcal{U}(t;s)$, which is defined in
\eqref{eq:def_mathcalU}, around the quintic Hartree dynamics. Then
one can utilize the evolution of the coherent state to estimate the
fluctuations for the dynamics of the factorized state. A technical
difficulty here was overcome by using the method in Rodnianski and
Schlein \cite{Rodnianski2009}, which is equivalent to Lemma \ref{lem:coherent_all}
in this paper. It was possible to overcome the difficulty by controlling
the fluctuation $\mathcal{U}(t;s)$ first by comparing it with an
approximate dynamics $\mathcal{\widetilde{U}}(t;s)$, whose generator
is $\widetilde{\mathcal{L}}(t)$ (see \ref{eq:def_mathcaltildeU}).
The idea was introduced, for two-body interaction, by Ginibre and
Velo \cite{Ginibre1979_1} as a limiting dynamics.

The main difficulty of this paper arise from three-body interaction.
Since we are dealing with three-body interaction
\[
V(x-y,x-z)=\lambda\left(v(x-y)v(x-z)+v(y-z)v(y-x)+v(z-x)v(z-y)\right)
\]
with $v(\cdot):=1/|x|$, the interaction $V$ is twice more singular
than Coulomb potential. Hence, we regularize each $v$ so that the
fluctuation of $\mathcal{\widetilde{U}}(t;s)$. For this, we need
to prove the wellposedness of $\varphi_{t}$ in $H^{1}(\mathbb{R}^{3})$
and obtain the global (in time) bound of $\|\varphi_{t}\|_{H^{1}(\mathbb{R}^{3})}$.
For that, we generalize the Hardy-Littlewood-Sobolev inequality in
Section \ref{sec:properties_quntic_Hartree}.

In this paper, we follow the approaches employed in \cite{Lee2013,Rodnianski2009}.
Since we want to use the approach for three-body interaction potential,
we regularize interaction potential by using $\alpha_{N}\leq N^{-\eta}$
for some $1\leq\eta<3/2$. By applying and extending the techniques
developed in \cite{ChenLeeSchlein2011,Lee2013,Rodnianski2009} for
the regularized potential, we obtain the optimal factor of order $N^{-1}$.
Here, the regularization has a key role, controlling the time evolution
due to the three-body interaction in the Fock space. While this gives
us the optimal bound for Fock states, it turns out that we lose the
rate of convergence between the original evolution and the regularized
evolution. 
Similar technique can be found in the BBGKY hierarchy approach by Chen and Holmer\cite{ChenHolmer2016IMRN,ChenHolmer2016JEMS} and in the physical space in the Fock space approach by Lewin, Nam, and Rougerie\cite{LewinNamRougerie2016}. Another version of such behavior can also be found in \cite{Lewin2015}.
To overcome this difficulty and to obtain the rate of convergence
$O(N^{-(1+4a)/(3+2a)})$ or the optimal rate of convergence $O(N^{-1})$,
we assume the initial data $\varphi_{0}$ to be in $H^{(3/2)+a}(\mathbb{R}^{3})$
or $H^{(5/2)+\varepsilon}(\mathbb{R}^{3})$, respectively.

This paper is organized as follows: First, in Section \ref{sec:Reg-potential},
we introduce the idea and strategy of the proof. Since we are using
regularized potential, we provide more details about it. Then, we
will compare the time evolution between before and after regularization.
Note that we are considering a very singular potential. We prove our
main theorem by proving the theorem for regularized potential in Section
\ref{sec:Pf-of-Main-Thm}. The main strategy is to embed our state
into the Fock space and compare the time evolution of our stated and
coherent state. For that, we use Propositions \ref{prop:Et1} and
\ref{prop:Et2}, which will be proved in Section \ref{sec:Pf-of-Props}.
In Section \ref{sec:comparison}, we prepare lemmas describing comparison
dynamics to prove Propositions \ref{prop:Et1} and \ref{prop:Et2}.
The lemmas are similar to the lemmas in previous works, for example
see \cite{ChenLeeSchlein2011,Lee2013,Rodnianski2009}. Since it has
been well-known about Fock space through out many papers, we review
bosonic Fock space formalism in Appendix \ref{sec:Fock_space}. In
Appendix \ref{sec:properties_quntic_Hartree}, we provide basic properties
of the solution of quintic Hartree equation to bound $\int\mathrm{d}x\mathrm{d}y\mathrm{d}z\,|V(x-y,x-z)|^{2}|\varphi_{t}(x)|^{2}|\varphi_{t}(y)|^{2}$.

\section{Regularization  of the interaction potential}\label{sec:Reg-potential}
We are going to use the Fock space approach which has been well established
for few decades (See, for example, \cite{ChenLeeSchlein2011,Lee2013,Rodnianski2009}).
Hence, we assume readers to be familiar to it. If one need more detail
about the Fock space, we provide it in Appendix \ref{sec:Fock_space}.
To utilize the Fock space approach for our case, we need to regularize initial data and use cut offed
potential. The following two subsections are introduced for those two regularizations.

For example, in Lemmas \ref{lem:NjU} and \ref{lem:L5},
we face the singularity of our potential. Since $\|V(x-y,x-\cdot)\varphi_{t}(\cdot)\|_{L^{2}}$
is singular in both $x$ and $y$, one should avoid or remove this
singularity. Thus, we detour the problem caused by singularity by
using regularized potential. Then, for the time evolution with regularized
potential, in Proposition \ref{prop:reged-V}, we obtain the optimal
rate of convergence with its mean-field dynamics with $H^{1}$-initial
data. The distance between of the original time evolution and the
time evolution with regularized potential will be provided as follows.

As we have talked in Section \ref{sec:introduction}, for the Fock
space analysis, the potential $V$ is more singular than we can utilize
the property of the solution of quintic Hartree equation. Hence, we
are going to remove the singularity of the potential $V$.

Let 
\begin{equation}
\overline{V}(x-y,x-z)=\lambda\left(\overline{v}(x-y)\overline{v}(x-z)+\overline{v}(y-z)\overline{v}(y-x)+\overline{v}(z-x)\overline{v}(z-y)\right)\label{eq:regularizedV}
\end{equation}
with
\[
\overline{v}(x-y)=\operatorname{sgn}\left(v(x-y)\right)\min\left\{ |v(x-y)|,\alpha_{N}^{-1}\right\} 
\]
where $\operatorname{sgn}(v(x))$ denotes the sign of $v(x)$. We
also define the regularized Hamiltonian
\begin{equation}
\overline{H}_{N}=\sum_{j=1}^{N}(-\Delta_{x_{j}})-\frac{1}{N^{2}}\sum_{i<j<k}\overline{V}(x_{i}-x_{j},x_{i}-x_{k}).\label{eq:regularizedHamiltonian}
\end{equation}
Then we have the following proposition for regularized Hamiltonian,
which give the optimal rate of convergence.
\begin{prop}
	\label{prop:reged-V}Let $\overline{V}(x-y,x-z)$ as in \eqref{eq:regularizedV}
	with $\alpha_{N}=N^{-\eta}$ for $1\leq\eta<3/2$. Let $\varphi_{0}\in H^{1}(\mathbb{R}^{3})$
	with $\|\varphi_{0}\|=1$, and $\varphi_{t}$ be the solution of the
	quintic Hartree equation
	\[
	\mathrm{i}\partial_{t}\phi_{t}=-\Delta\phi_{t}+\frac{1}{2}\left(\int\mathrm{d}y\mathrm{d}z\,\overline{V}(x-y,x-z)|\phi_{t}(y)|^{2}|\phi_{t}(z)|^{2}\right)\phi_{t}
	\]
	with initial data $\varphi_{t=0}=\varphi_{0}$. Let $\overline{\psi}_{N,t}=e^{-\mathrm{i}\overline{H}_{N}t}\varphi_{0}^{\otimes N}$
	and $\overline{\gamma}_{N,t}^{(1)}$ be the one-particle reduced density
	associated with $\overline{\psi}_{N,t}$. Then there exist constants
	$C$ and $K$, depending only on $\|\varphi\|_{H^{1}}$, $D,$and
	$\lambda$ such that
	\[
	\operatorname{Tr}\left|\overline{\gamma}_{N,t}^{(1)}-|\phi_{t}\rangle\langle\phi_{t}|\right|\leq\frac{Ce^{Kt}}{N^{2-\eta}}.
	\]
\end{prop}

The proof of Proposition \ref{prop:reged-V} is given in Section \ref{sec:Pf-of-Main-Thm}.
It makes use of a representation of the problem on the bosonic Fock
space, detail for the Fock space, see Appendix \ref{sec:Fock_space}.
For now, want to argue that the evolution governed by $H_{N}$ and
$\overline{H}_{N}$ are similar enough.
\begin{lem}
	\label{lem:psi-psihat-L2-dist}Let $\psi_{N,0}=\varphi_{0}^{\otimes N},$
	$\varphi_{0}\in H^{(3/2)+a}(\mathbb{R}^{3})$, and $\|\varphi_{0}\|_{L^{2}(\mathbb{R}^{3})}=1$.
	Let $\psi_{N,t}=e^{-\mathrm{i}H_{N}t}\varphi_{0}^{\otimes N}$ and
	$\overline{\psi}_{N,t}=e^{-\mathrm{i}\overline{H}_{N}t}\psi_{N,0}$.
	Then there exist a constant $C>0$ such that
	\begin{enumerate}
		\item If $0\leq a<1/2$, then
		\[
		\|\psi_{N,t}-\overline{\psi}_{N,t}\|^{2}\leq CN\alpha_{N}^{1+2a}|t|\quad\text{and}
		\]
		\item If $a>1$, then
		\[
		\|\psi_{N,t}-\overline{\psi}_{N,t}\|^{2}\leq CN\alpha_{N}^{3}|t|\quad\text{and}
		\]
	\end{enumerate}
	for all $N\in\mathbb{N}$, $t\in\mathbb{R}$.
\end{lem}

\begin{proof}
	We consider the derivative
	\begin{align*}
	\frac{\mathrm{d}}{\mathrm{d}t}\|\psi_{N,t}-\overline{\psi}_{N,t}\|^{2} & =-2\operatorname{Re}\frac{\mathrm{d}}{\mathrm{d}t}\langle\psi_{N,t},\overline{\psi}_{N,t}\rangle\\
	& =2\operatorname{Im}\langle(H_{N}-\overline{H}_{N})\psi_{N,t},\overline{\psi}_{N,t}\rangle\\
	& =\frac{2}{N^{2}}\sum_{i<j<k}^{N}\operatorname{Im}\langle(V(x_{i}-x_{j},x_{i}-x_{k})-\overline{V}(x_{i}-x_{j},x_{i}-x_{k}))\psi_{N,t},\overline{\psi}_{N,t}\rangle.
	\end{align*}
	Observe that \eqref{eq:regularizedV} gives us that
	\begin{align*}
	&\left|V(x-y,x-z)-\overline{V}(x-y,x-z)\right|\\
	& \leq|\lambda|\Bigg(\left|v(x-y)v(x-z)-\overline{v}(x-y)\overline{v}(x-z)\right|
	+\left|v(y-x)v(y-z)-\overline{v}(y-z)\overline{v}(y-x)\right|\\
	& \qquad\qquad
	+\left|v(z-x)v(z-y)-\overline{v}(z-x)\overline{v}(z-y)\right|\Bigg)\\
	& \leq|\lambda|\Bigg(\left|v(x-y)v(x-z)-v(x-y)\overline{v}(x-z)\right|+\left|v(x-y)\overline{v}(x-z)-\overline{v}(x-y)\overline{v}(x-z)\right|\\
	& \qquad\qquad+\left|v(y-z)v(y-x)-v(y-z)\overline{v}(y-x)\right|+\left|v(y-z)\overline{v}(y-x)-\overline{v}(y-z)\overline{v}(y-x)\right|\\
	& \qquad\qquad+\left|v(z-x)v(z-y)-v(z-x)\overline{v}(z-y)\right|+\left|v(z-x)\overline{v}(z-y)-\overline{v}(z-x)\overline{v}(z-y)\right|\Bigg)\\
	& \leq|\lambda|\Bigg(|v(x-y)|\left|v(x-z)-\overline{v}(x-z)\right|+\left|v(x-y)-\overline{v}(x-y)\right||\overline{v}(x-z)|\\
	& \qquad\qquad+|v(y-z)|\left|v(y-x)-\overline{v}(y-x)\right|+\left|v(y-z)-\overline{v}(y-z)\right||\overline{v}(y-x)|\\
	& \qquad\qquad+|v(z-x)|\left|v(z-y)-\overline{v}(z-y)\right|+\left|v(z-x)-\overline{v}(z-x)\right||\overline{v}(z-y)|\Bigg).
	\end{align*}
	Note that, for any $a\geq0$,
	
	\begin{equation}
	\left|v-\overline{v}\right|\leq|v|\cdot\mathbf{1}(|v|\geq\alpha_{N}^{-1})\leq|v|^{2+2a}\alpha_{N}^{1+2a}.\label{eq:v-vhat}
	\end{equation}
	Note that
	\begin{align*}
	&\left|V(x-y,x-z)-\overline{V}(x-y,x-z)\right|\\
	& \leq2|\lambda|\alpha_{N}^{1+2a}\left(|v(x-y)||v(x-z)|^{2+2a}+|v(y-z)||v(y-x)|^{2+2a}+|v(z-x)||v(z-y)|^{2+2a}\right).
	\end{align*}
	Now, we set any $\varepsilon:=a-1>0$. Using Hardy
	inequality and Sobolev embedding, we obtain
	\begin{align*}
	& \int_{\mathbb{R}^{3}}\mathrm{d}x\,\left|v(x)-\overline{v}(x)\right||\varphi(x)|^{2}=\int_{B(0,\alpha_{N})}\mathrm{d}x\,|v(x)|\,|\varphi(x)|^{2}\\
	&\leq C\int_{B(0,\alpha_{N})}\mathrm{d}x\,|(1-\Delta)^{1/4}\varphi(x)|^{2}
	\leq C\int_{\mathbb{R}^{3}}\mathrm{d}x\,|(1-\Delta)^{1/4}\varphi(x)|^{2}\chi_{B(0,\alpha_{N})}\\
	&\leq C\|(1-\Delta)^{1/4}\varphi\|_{\infty}^{2}\int_{\mathbb{R}^{3}}\mathrm{d}x\,\chi_{B(0,\alpha_{N})}\leq C\|(1-\Delta)^{1/4}\varphi\|_{H^{(3/2)+\varepsilon}}^{2}\alpha_{N}^{3}\\
	&\leq C\|\varphi\|_{H^{2+\varepsilon}}^{2}\alpha_{N}^{3}.
	\end{align*}
	In short, we have
	\begin{equation}
	\left|v-\overline{v}\right|\leq C(1-\Delta)^{2+\varepsilon}\alpha_{N}^{3}.\label{eq:v-vhat-2}
	\end{equation}
	Then
	\begin{align*}
	& \left|\frac{\mathrm{d}}{\mathrm{d}t}\|\psi_{N,t}-\overline{\psi}_{N,t}\|^{2}\right|\\
	& \leq CN\left|\langle(V(x-y,x-z)-\overline{V}(x-y,x-z))\psi_{N,t},\overline{\psi}_{N,t}\rangle\right|\\
	& \leq CN\alpha_{N}^{3}\Bigg(\langle\psi_{N,t},(1-\Delta_{y})^{1/2}(1-\Delta_{z})^{2+\varepsilon}\psi_{N,t}\rangle^{1/2}\langle\overline{\psi}_{N,t},(1-\Delta_{y})^{1/2}(1-\Delta_{z})^{2+\varepsilon}\overline{\psi}_{N,t}\rangle^{1/2}\\
	& \qquad\qquad+\langle\psi_{N,t},(1-\Delta_{x})^{2+\varepsilon}(1-\Delta_{z})^{1/2}\psi_{N,t}\rangle^{1/2}\langle\overline{\psi}_{N,t},(1-\Delta_{x})^{2}(1-\Delta_{z})^{1/2}\overline{\psi}_{N,t}\rangle^{1/2}\\
	& \qquad\qquad+\langle\psi_{N,t},(1-\Delta_{x})^{1/2}(1-\Delta_{y})^{2+\varepsilon}\psi_{N,t}\rangle^{1/2}\langle\overline{\psi}_{N,t},(1-\Delta_{x})^{1/2}(1-\Delta_{y})^{2+\varepsilon}\overline{\psi}_{N,t}\rangle^{1/2}\Bigg).
	\end{align*}
	Note that
	\begin{align*}
	&N^{(5/2)+\varepsilon}\langle\psi_{N,t},(1-\Delta_{y})^{1/2}(1-\Delta_{z})^{2+\varepsilon}\psi_{N,t}\rangle 
	\leq C\langle\psi_{N,t},(H_{N}+N)^{(5/2)+\varepsilon}\psi_{N,t}\rangle\\
	& \leq C\langle\varphi_{0}^{\otimes N},(H_{N}+N)^{(5/2)+\varepsilon}\varphi_{0}^{\otimes N}\rangle \leq CN^{(5/2)+\varepsilon}\|\varphi_{0}\|_{H^{(5/2)+\varepsilon}}^{2}.
	\end{align*}
	Therefore, we have that
	\[
	\frac{\mathrm{d}}{\mathrm{d}t}\|\psi_{N,t}-\overline{\psi}_{N,t}\|^{2}\leq CN\alpha_{N}^{3}
	\]
	where the constant $C$ depends on $\|\varphi_{0}\|_{H^{(5/2)+\varepsilon}}=\|\varphi_{0}\|_{H^{(3/2)+a}}$.
	
	This concludes the proof of the desired lemma.
\end{proof}
\begin{cor}
	\label{cor:gamma-gammahat}Let $\gamma_{N,t}^{(k)}$ and $\overline{\gamma}_{N,t}^{(k)}$
be the $k$-particle reduced densities associated with $\psi_{N,t}=e^{-\mathrm{i}H_{N}t}\varphi_{0}^{\otimes N}$
and $\overline{\psi}_{N,t}=e^{-\mathrm{i}\overline{H}_{N}t}\psi_{N,0}$.
Suppose $\varphi\in H^{(3/2)+a}(\mathbb{R}^{3})$ and $\alpha_{N}=N^{-\eta}$.
Then there exist a constant $C>0$ such that
\begin{enumerate}
	\item If $0\leq a<1/2$, then
	\[
	\operatorname{Tr}\left|\gamma_{N,t}^{(k)}-\overline{\gamma}_{N,t}^{(k)}\right|\leq CN^{1-\eta(1+2a)}|t|^{1/2}.
	\]
	\item If $a>1$, then
	\[
	\operatorname{Tr}\left|\gamma_{N,t}^{(k)}-\overline{\gamma}_{N,t}^{(k)}\right|\leq CN^{1-3\eta}|t|^{1/2}.
	\]
\end{enumerate}
\end{cor}

\begin{proof}
	See \cite[Corollary 2.1]{ChenLeeSchlein2011}.
\end{proof}

\begin{lem}
	Let $\varphi_{t}$ be the solution of the quintic Hartree equation
	(\ref{eq:quinticHartree}) and $\phi_{t}$ the solution
	of the quintic Hartree equation
	\[
	\mathrm{i}\partial_{t}\phi_{t}=-\Delta\phi_{t}+\frac{1}{2}\left(\int\mathrm{d}y\mathrm{d}z\,\overline{V}(x,y,z)|\phi_{t}(y)|^{2}|\phi_{t}(z)|^{2}\right)\phi_{t}
	\]
	with regularized potential $\overline{V}$ with the same initial data
	$\varphi_{t=0}=\phi_{t=0}=\varphi_{0}\in H^{(3/2)+a}(\mathbb{R}^{3})$.
	Suppose $\alpha_{N}\leq N^{-\eta}$. Then
	\begin{enumerate}
		\item If $0\leq a<1/2$, then
		\begin{equation}
			\|\varphi_{t}-\phi_{t}\|\leq CN^{(1-\eta(1+2a))/2}e^{K|t|},\label{eq:phi-phihat-L2}
		\end{equation}
		Therefore
		\begin{equation}
			\operatorname{Tr}\left||\varphi_{t}\rangle\langle\varphi_{t}|^{\otimes k}-|\phi_{t}\rangle\langle\phi_{t}|^{\otimes k}\right|\leq CkN^{(1-\eta(1+2a))/2}e^{K|t|}\label{eq:phi-phihat-tr}
		\end{equation}
		\item If $a>1$, then
		\begin{equation}
			\|\varphi_{t}-\phi_{t}\|\leq CN^{(1-3\eta)/2}e^{K|t|}\label{eq:phi-phihat-L2-2}
		\end{equation}
		Therefore
		\begin{equation}
			\operatorname{Tr}\left||\varphi_{t}\rangle\langle\varphi_{t}|^{\otimes k}-|\phi_{t}\rangle\langle\phi_{t}|^{\otimes k}\right|\leq CkN^{(1-3\eta)/2}e^{K|t|}\label{eq:phi-phihat-tr-1}
		\end{equation}
	\end{enumerate}
	for any $k\in\mathbb{N}$.
\end{lem}

\begin{proof}
	From Lemma \ref{lem:H1bdd} we see that $\|\varphi_{t}\|_{H^{1}},\|\phi_{t}\|_{H^{1}}<C$,
	for some constant $C$ which only depends on $\|\varphi_{0}\|_{H^{1}}$.
	We calculate
	\begin{align*}
	&\frac{\mathrm{d}}{\mathrm{d}t}\|\varphi_{t}-\phi_{t}\|^{2} \\
	& =2\operatorname{Im}\left\langle \varphi_{t},[\int\mathrm{d}y\mathrm{d}z\,V(x-y,x-z)|\varphi_{t}(y)|^{2}|\varphi_{t}(z)|^{2}-\int\mathrm{d}y\mathrm{d}z\,\overline{V}(x-y,x-z)|\phi_{t}(y)|^{2}|\phi_{t}(z)|^{2}]\overline{\varphi_{t}}\right\rangle \\
	& =2\operatorname{Im}\left\langle \varphi_{t},[\int\mathrm{d}y\mathrm{d}z\,\left(V(x-y,x-z)-\overline{V}(x-y,x-z)\right)|\varphi_{t}(y)|^{2}|\varphi_{t}(z)|^{2}]\overline{\varphi_{t}}\right\rangle \\
	& \qquad+=2\operatorname{Im}\left\langle \varphi_{t},[\int\mathrm{d}y\mathrm{d}z\,\overline{V}(x-y,x-z)\left(|\varphi_{t}(y)|^{2}|\varphi_{t}(z)|^{2}-|\phi_{t}(y)|^{2}|\phi_{t}(z)|^{2}\right)](\overline{\varphi_{t}}-\varphi_{t})\right\rangle .
	\end{align*}
	Then, using \eqref{eq:v-vhat}, we obtain
	\begin{align}
	& \left|\frac{\mathrm{d}}{\mathrm{d}t}\|\varphi_{t}-\phi_{t}\|^{2}\right|\nonumber \\
	& \leq2\alpha_{N}^{2}\|\varphi_{t}\|\|\phi_{t}\|\sup_{x}\int\mathrm{d}y,\mathrm{d}z\,|v(x-y)|^{2}|v(x-z)|^{2}|\varphi_{t}(y)|^{2}|\varphi_{t}(z)|^{2}\nonumber \\
	& \quad+2\|\overline{\varphi_{t}}-\varphi_{t}\|\|\varphi_{t}\|\sup_{x}\int\mathrm{d}y\mathrm{d}z\,|V(x-y,x-z)|\left||\varphi_{t}(y)\varphi_{t}(z)|-|\phi_{t}(y)\phi_{t}(z)|\right| \nonumber\\
	&\qquad\qquad\qquad\qquad\qquad\qquad\qquad\times\left(|\varphi_{t}(y)||\varphi_{t}(z)|+|\phi_{t}(y)||\phi_{t}(z)|\right)\nonumber \\
	& \leq C\alpha_{N}^{2}+C\|\overline{\varphi_{t}}-\varphi_{t}\|^{2}.\label{eq:ddt-phi-phit-sq}
	\end{align}
	In the last inequality, we have used that
	\begin{align*}
	& \int\mathrm{d}y\mathrm{d}z\,|V(x-y,x-z)|\left||\varphi_{t}(y)\varphi_{t}(z)|-\phi_{t}(y)\phi_{t}(z)\right|\left(|\varphi_{t}(y)||\varphi_{t}(z)|+|\phi_{t}(y)||\phi_{t}(z)|\right)\\
	& \leq C\left(\int\mathrm{d}y\mathrm{d}z\,|\varphi_{t}(y)\varphi_{t}(z)-\phi_{t}(y)\phi_{t}(z)|^{2}\right)^{1/2}\\
	&\times\left(\int\mathrm{d}y\mathrm{d}z\,|V(x-y,x-z)|^{2}\left(|\varphi_{t}(y)|^{2}|\varphi_{t}(z)|^{2}+|\phi_{t}(y)|^{2}|\phi_{t}(z)|^{2}\right)\right)^{1/2}\\
	& \leq C\|\overline{\varphi_{t}}-\varphi_{t}\|.
	\end{align*}
	From (\ref{eq:ddt-phi-phit-sq}) we obtain by Grönwall inequality
\[
\|\varphi_{t}-\phi_{t}\|^{2}\leq CN^{(1-\eta(1+2a))/2}(e^{C|t|}-1)
\]
if $0\leq a<1/2$, and 
\[
\|\varphi_{t}-\phi_{t}\|^{2}\leq CN^{(1-3\eta)/2}(e^{C|t|}-1)
\]
if $a>1$.
	This concludes the proof by following the proof of \cite[Lemma 2.2]{ChenLeeSchlein2011}.
\end{proof}

\section{Proof of the main result\label{sec:Pf-of-Main-Thm}}

Now, we are ready prove the main result of the paper, Theorem \ref{thm:main}.
To have that, we first prove the Proposition \ref{prop:reged-V}.

\subsection{Unitary operators and their generators}

As we define bosonic Fock space in Section \ref{sec:Fock_space},
the new Hamiltonian for the Fock space evolution can be written as
\begin{equation}
\mathcal{H}_{N}=\int\mathrm{d}x\,a_{x}^{*}(-\Delta)a_{x}+\frac{1}{6N^{2}}\int\mathrm{d}x\mathrm{d}y\mathrm{d}z\,V(x-y,x-z)a_{x}^{*}a_{y}^{*}a_{z}^{*}a_{z}a_{y}a_{x}.\label{eq:Fock_space_Hamiltonian-1}
\end{equation}
Since we have $(\mathcal{H}_{N}\psi)^{(N)}=H_{N}\psi^{(N)}$ for $\psi\in\mathcal{F}$,
\eqref{eq:Fock_space_Hamiltonian-1} can be justified as a proper generalization
of \eqref{eq:N_body_Hamiltonian}. Since we are going to use regularized
potential, we define 
\begin{equation}
\overline{\mathcal{H}}_{N}=\int\mathrm{d}x\,a_{x}^{*}(-\Delta)a_{x}+\frac{1}{6N^{2}}\int\mathrm{d}x\mathrm{d}y\mathrm{d}z\,\overline{V}(x-y,x-z)a_{x}^{*}a_{y}^{*}a_{z}^{*}a_{z}a_{y}a_{x}\label{eq:Fock_space_Hamiltonian}
\end{equation}
which is also a generalization of $\overline{H}_{N}$ for the Fock
space.

The one-particle marginal density $\gamma_{\psi}^{(1)}$ associated
with $\psi$ is 
\begin{equation}
\gamma_{\psi}^{(1)}\left(x;y\right)=\frac{1}{\left\langle \psi,\mathcal{N}\psi\right\rangle }\left\langle \psi,a_{y}^{*}a_{x}\psi\right\rangle .\label{eq:Kernel_gamma}
\end{equation}
Note that $\gamma_{\psi}^{(1)}$ is a trace class operator on $L^{2}(\mathbb{R}^{3})$
and $\text{Tr }\gamma_{\psi}^{(1)}=1$.

Let $\overline{\gamma}_{N,t}^{(1)}$ be the kernel of the one-particle
marginal density associated with the time evolution of the factorized
state $\varphi_{0}^{\otimes N}$ for Hamiltonian $\overline{\mathcal{H}}_{N}$.
By definition, 
\begin{align}
\overline{\gamma}_{N,t}^{(1)} & =\frac{\left\langle e^{-\mathrm{i}\overline{\mathcal{H}}_{N}t}\varphi^{\otimes N},a_{y}^{*}a_{x}e^{-\mathrm{i}\overline{\mathcal{H}}_{N}t}\varphi^{\otimes N}\right\rangle }{\left\langle e^{-\mathrm{i}\overline{\mathcal{H}}_{N}t}\varphi^{\otimes N},\mathcal{N}e^{-\mathrm{i}\overline{\mathcal{H}}_{N}t}\varphi^{\otimes N}\right\rangle }=\frac{1}{N}\left\langle \varphi^{\otimes N},e^{\mathrm{i}\overline{\mathcal{H}}_{N}t}a_{y}^{*}a_{x}e^{-\mathrm{i}\overline{\mathcal{H}}_{N}t}\varphi^{\otimes N}\right\rangle \nonumber \\
& =\frac{1}{N}\left\langle \frac{\left(a^{*}(\varphi)\right)^{N}}{\sqrt{N!}}\Omega,e^{\mathrm{i}\overline{\mathcal{H}}_{N}t}a_{y}^{*}a_{x}e^{-\mathrm{i}\overline{\mathcal{H}}_{N}t}\frac{\left(a^{*}(\varphi)\right)^{N}}{\sqrt{N!}}\Omega\right\rangle .\label{eq:marginal_factorized}
\end{align}
If we put the coherent states instead of the factorized initial data
in \eqref{eq:marginal_factorized} and expand $a_{y}^{*}a_{x}$ around
$N\overline{\varphi_{t}(y)}\varphi_{t}(x)$, then it is enough to
consider the operator 
\begin{align}
& W^{*}(\sqrt{N}\varphi_{s})e^{\mathrm{i}\overline{\mathcal{H}}_{N}\left(t-s\right)}(a_{x}-\sqrt{N}\varphi_{t}(x))e^{-\mathrm{i}\overline{\mathcal{H}}_{N}\left(t-s\right)}W(\sqrt{N}\varphi_{s})\label{eq:introducingU}\\
& =W^{*}(\sqrt{N}\varphi_{s})e^{\mathrm{i}\overline{\mathcal{H}}_{N}\left(t-s\right)}W(\sqrt{N}\varphi_{t})a_{x}W^{*}(\sqrt{N}\varphi_{t})e^{-\mathrm{i}\overline{\mathcal{H}}_{N}\left(t-s\right)}W(\sqrt{N}\varphi_{s}).\nonumber 
\end{align}
Now we are lead to understand the operator $W^{*}(\sqrt{N}\varphi_{t})e^{-\mathrm{i}\overline{\mathcal{H}}_{N}\left(t-s\right)}W(\sqrt{N}\varphi_{s})$.
For the understanding, since we know that for $t=s$, we investigate
the time evolution of the operator by differentiate it with respect
to $t$. One can compute directly such that
\begin{equation}
\mathrm{i}\partial_{t}W^{*}(\sqrt{N}\varphi_{t})e^{-\mathrm{i}\overline{\mathcal{H}}_{N}\left(t-s\right)}W(\sqrt{N}\varphi_{s})=\mathcal{L}W^{*}(\sqrt{N}\varphi_{t})e^{-\mathrm{i}\overline{\mathcal{H}}_{N}\left(t-s\right)}W(\sqrt{N}\varphi_{s}),\label{eq:derivative_decomposition}
\end{equation}
where $\mathcal{L}:=\sum_{k=0}^{6}\mathcal{L}_{k}(t)$ and the exact
formulas for $\mathcal{L}_{k}$ are as follows:

We consider evolution
\begin{equation}
\mathrm{i}\partial_{t}\mathcal{U}=\mathcal{L}\,\mathcal{U}\label{eq:def_mathcalU}
\end{equation}
with
\[
\mathcal{L}=\mathcal{L}_{0}+\mathcal{L}_{1}+\mathcal{L}_{2}+\mathcal{L}_{3}+\mathcal{L}_{4}+\mathcal{L}_{5}+\mathcal{L}_{6}
\]
where
\begin{align*}
\mathcal{L}_{0} & =\frac{N}{6}\int\mathrm{d}x\mathrm{d}y\mathrm{d}z\,\overline{V}(x-y,x-z)|\varphi_{t}(x)|^{2}|\varphi_{t}(y)|^{2}|\varphi_{t}(z)|^{2}\\
\mathcal{L}_{1} & =0\\
\mathcal{L}_{2} & =\int\mathrm{d}x\,a_{x}^{*}(-\Delta)a_{x}\\
& \qquad+\frac{1}{6}\int\mathrm{d}x\mathrm{d}y\mathrm{d}z\,\overline{V}(x-y,x-z)\Bigg[3|\varphi_{t}(x)|^{2}\left(\varphi_{t}(y)\varphi_{t}(z)a_{y}^{*}a_{z}^{*}+\overline{\varphi_{t}(y)}\overline{\varphi_{t}(z)}a_{z}a_{y}\right)\\
& \qquad\qquad\qquad\qquad\qquad\qquad+3|\varphi_{t}(x)|^{2}\left(\varphi_{t}(y)\overline{\varphi_{t}(z)}a_{y}^{*}a_{z}+\overline{\varphi_{t}(y)}\varphi_{t}(z)a_{z}^{*}a_{y}\right)\Bigg]\\
\mathcal{L}_{3} & =\frac{1}{6\sqrt{N}}\int\mathrm{d}x\mathrm{d}y\mathrm{d}z\,\overline{V}(x-y,x-z)\Bigg[\left(\varphi_{t}(x)\varphi_{t}(y)\varphi_{t}(z)a_{x}^{*}a_{y}^{*}a_{z}^{*}+\overline{\varphi_{t}(x)}\overline{\varphi_{t}(y)}\overline{\varphi_{t}(z)}a_{x}a_{y}a_{z}\right)\\
& \qquad\qquad\qquad\qquad\qquad+3\left(\varphi_{t}(x)\varphi_{t}(y)\overline{\varphi_{t}(z)}a_{x}^{*}a_{y}^{*}a_{z}+\overline{\varphi_{t}(x)}\overline{\varphi_{t}(y)}\varphi_{t}(z)a_{z}^{*}a_{x}a_{y}\right)\Bigg]\\
\mathcal{L}_{4}^{c} & =\frac{1}{6N}\int\mathrm{d}x\mathrm{d}y\mathrm{d}z\,\overline{V}(x-y,x-z)\Bigg[3\left(\varphi_{t}(x)\varphi_{t}(y)a_{x}^{*}a_{y}^{*}a_{z}^{*}a_{z}+\overline{\varphi_{t}(x)}\overline{\varphi_{t}(y)}a_{x}a_{y}a_{z}^{*}a_{z}\right)\Bigg]\\
\mathcal{L}_{4}^{r} & =\frac{1}{6N}\int\mathrm{d}x\mathrm{d}y\mathrm{d}z\,\overline{V}(x-y,x-z)\Bigg[6\varphi_{t}(x)\overline{\varphi_{t}(y)}a_{x}^{*}a_{z}^{*}a_{z}a_{y}+3|\varphi_{t}(x)|^{2}a_{y}^{*}a_{z}^{*}a_{z}a_{y}\Bigg]\\
\mathcal{L}_{5} & =\frac{1}{2N\sqrt{N}}\int\mathrm{d}x\mathrm{d}y\mathrm{d}z\,\overline{V}(x-y,x-z)a_{x}^{*}a_{y}^{*}\left(\varphi_{t}(z)a_{z}^{*}+\overline{\varphi_{t}(z)}a_{z}\right)a_{y}a_{x}\\
\mathcal{L}_{6} & =\frac{1}{6N^{2}}\int\mathrm{d}x\mathrm{d}y\mathrm{d}z\,\overline{V}(x-y,x-z)a_{x}^{*}a_{y}^{*}a_{z}^{*}a_{z}a_{y}a_{x}.
\end{align*}

\subsection{Proof of Theorem \ref{thm:main}}

As explained in Section \ref{sec:introduction}, we use the technique
developed in \cite{ChenLeeSchlein2011,XChen2012,Lee2013,Rodnianski2009}.
The proof of Theorem \ref{thm:main} is a consequence of Corollary
\ref{cor:gamma-gammahat} and Proposition \ref{prop:reged-V}. 

Let $• \overline{\varphi}_{t}$ be the solution of cut-offed Hartree equation. Then we have
\begin{align*}
\Tr\Big|\gamma_{N}^{(1)}-|\varphi_{t}\rangle\langle\varphi_{t}|\Big|&\leq\Tr\Big|\gamma_{N}^{(1)}-\overline{\gamma}_{N}^{(1)}\Big|+\Tr\Big|\overline{\gamma}_{N}^{(1)}-|\overline{\varphi}_{t}\rangle\langle\overline{\varphi}_{t}|\Big|+\Tr\Big||\overline{\varphi}_{t}\rangle\langle\overline{\varphi}_{t}|-|\varphi_{t}\rangle\langle\varphi_{t}|\Big|
\end{align*}

Hence,
it is enough to prove Proposition \ref{prop:reged-V}. For $\varphi\in H^{(3/2)+a}(\mathbb{R}^{3})$,
we put $\eta=5/4$ when $0\leq a<1/2$, and we put $\eta=1$ when
$a>1$. 
This conclude the main theorem. The proof of Proposition \ref{prop:reged-V}
consists of the following two propositions. Recall the definition
of $d_{N}$ in \eqref{eq:d_N}. In this section, we use $\mathcal{U}(t)$
instead of $\mathcal{U}(t;s)$ for notional simplicity.
\begin{prop}
	\label{prop:Et1} Suppose that the assumptions in Theorem \ref{thm:main}
	hold. For a Hermitian operator $J$ on $H^{1}(\mathbb{R}^{3})$, let
	\[
	E_{t}^{1}(J):=\frac{d_{N}}{N}\left\langle W^{*}(\sqrt{N}\varphi)\frac{(a^{*}(\varphi))^{N}}{\sqrt{N!}}\Omega,\mathcal{U}^{*}(t)d\Gamma(J)\mathcal{U}(t)\Omega\right\rangle 
	\]
	Then, there exist constants $C$ and $K$, depending only on $\lambda$
	and $\|\varphi_{0}\|_{H^{1}}$, such that 
	\[
	\left|E_{t}^{1}(J)\right|\leq\frac{C\|J\|e^{Kt}}{N^{2-\eta}}.
	\]
\end{prop}

\begin{prop}
	\label{prop:Et2} Suppose that the assumptions in Theorem \ref{thm:main}
	hold. For a Hermitian operator $J$ on $H^{1}(\mathbb{R}^{3})$, let
	\[
	E_{t}^{2}(J):=\frac{d_{N}}{\sqrt{N}}\left\langle W^{*}(\sqrt{N}\varphi)\frac{(a^{*}(\varphi))^{N}}{\sqrt{N!}}\Omega,\mathcal{U}^{*}(t)\phi(J\varphi_{t})\mathcal{U}(t)\Omega\right\rangle 
	\]
	Then, there exist constants $C$ and $K$, depending only on $\lambda$
	and $\|\varphi_{0}\|_{H^{1}}$, such that 
	\[
	\left|E_{t}^{2}(J)\right|\leq\frac{C\|J\|e^{Kt}}{N^{2-\eta}}.
	\]
\end{prop}

Proof of Propositions \ref{prop:Et1} and \ref{prop:Et2} will be
given later in section \ref{sec:Pf-of-Props}. With Propositions \ref{prop:Et1}
and \ref{prop:Et2}, we now prove Proposition \ref{prop:reged-V}.
\begin{proof}[Proof of Proposition \ref{prop:reged-V}]
	Formally, the proof is the same with previous results but for the
	sake of completeness, we include the proof of this theorem. Recall
	that 
	\begin{equation}
	\overline{\gamma}_{N,t}^{(1)}=\frac{1}{N}\left\langle \frac{\left(a^{*}(\varphi)\right)^{N}}{\sqrt{N!}}\Omega,e^{i\overline{\mathcal{H}}_{N}t}a_{y}^{*}a_{x}e^{-i\overline{\mathcal{H}}_{N}t}\frac{\left(a^{*}(\varphi)\right)^{N}}{\sqrt{N!}}\Omega\right\rangle .
	\end{equation}
	From the definition of the creation operator in \eqref{eq:creation},
	we can easily find that 
	\begin{equation}
	\{0,0,\dots,0,\varphi^{\otimes N},0,\dots\}=\frac{\left(a^{*}(\varphi)\right)^{N}}{\sqrt{N!}}\Omega,\label{eq:coherent_vec}
	\end{equation}
	where the $\varphi^{\otimes N}$ on the left-hand side is in the $N$-th
	sector of the Fock space. Recall that $P_{N}$ is the projection onto
	the $N$-particle sector of the Fock space. From \eqref{Weyl_f},
	we find that 
	\[
	\frac{\left(a^{*}(\varphi)\right)^{N}}{\sqrt{N!}}\Omega=\frac{\sqrt{N!}}{N^{N/2}e^{-N/2}}P_{N}W(\sqrt{N}\varphi)\Omega=d_{N}P_{N}W(\sqrt{N}\varphi)\Omega.
	\]
	Since $\overline{\mathcal{H}}_{N}$ does not change the number of particles,
	we also have that 
	\begin{align*}
	\overline{\gamma}_{N,t}^{(1)}(x;y) & =\frac{1}{N}\left\langle \frac{\left(a^{*}(\varphi)\right)^{N}}{\sqrt{N!}}\Omega,e^{\mathrm{i}\overline{\mathcal{H}}_{N}t}a_{y}^{*}a_{x}e^{-\mathrm{i}\overline{\mathcal{H}}_{N}t}\frac{\left(a^{*}(\varphi)\right)^{N}}{\sqrt{N!}}\Omega\right\rangle \\
	& =\frac{d_{N}}{N}\left\langle \frac{\left(a^{*}(\varphi)\right)^{N}}{\sqrt{N!}}\Omega,e^{\mathrm{i}\overline{\mathcal{H}}_{N}t}a_{y}^{*}a_{x}e^{-\mathrm{i}\overline{\mathcal{H}}_{N}t}P_{N}W(\sqrt{N}\varphi)\Omega\right\rangle \\
	& =\frac{d_{N}}{N}\left\langle \frac{\left(a^{*}(\varphi)\right)^{N}}{\sqrt{N!}}\Omega,P_{N}e^{\mathrm{i}\overline{\mathcal{H}}_{N}t}a_{y}^{*}a_{x}e^{-\mathrm{i}\overline{\mathcal{H}}_{N}t}W(\sqrt{N}\varphi)\Omega\right\rangle \\
	& =\frac{d_{N}}{N}\left\langle \frac{\left(a^{*}(\varphi)\right)^{N}}{\sqrt{N!}}\Omega,e^{\mathrm{i}\overline{\mathcal{H}}_{N}t}a_{y}^{*}a_{x}e^{-\mathrm{i}\overline{\mathcal{H}}_{N}t}W(\sqrt{N}\varphi)\Omega\right\rangle .
	\end{align*}
	To simplify it further, we use the relation 
	\[
	e^{\mathrm{i}\overline{\mathcal{H}}_{N}t}a_{x}e^{-\mathrm{i}\overline{\mathcal{H}}_{N}t}=W(\sqrt{N}\varphi)\mathcal{U}^{*}(t)(a_{x}+\sqrt{N}\phi_{t}(x))\mathcal{U}(t)W^{*}(\sqrt{N}\varphi)
	\]
	(and an analogous result for the creation operator) to obtain that
	\begin{align*}
	&\overline{\gamma}_{N,t}^{(1)}(x;y)\\
	 & =\frac{d_{N}}{N}\left\langle \frac{\left(a^{*}(\varphi)\right)^{N}}{\sqrt{N!}}\Omega,e^{\mathrm{i}\overline{\mathcal{H}}_{N}t}a_{y}^{*}a_{x}e^{-\mathrm{i}\overline{\mathcal{H}}_{N}t}W(\sqrt{N}\varphi)\Omega\right\rangle \\
	& =\frac{d_{N}}{N}\left\langle \frac{\left(a^{*}(\varphi)\right)^{N}}{\sqrt{N!}}\Omega,W(\sqrt{N}\varphi)\mathcal{U}^{*}(t)(a_{y}^{*}+\sqrt{N}\,\overline{\phi_{t}\left(y\right)})(a_{x}+\sqrt{N}\phi_{t}(x))\mathcal{U}(t)\Omega\right\rangle.
	\intertext{Thus,}
	& =\frac{d_{N}}{N}\left\langle \frac{\left(a^{*}(\varphi)\right)^{N}}{\sqrt{N!}}\Omega,W(\sqrt{N}\varphi)\mathcal{U}^{*}(t)a_{y}^{*}a_{x}\mathcal{U}(t)\Omega\right\rangle \\
	& \quad+\overline{\varphi_{t}\left(y\right)}\frac{d_{N}}{\sqrt{N}}\left\langle \frac{\left(a^{*}(\varphi)\right)^{N}}{\sqrt{N!}}\Omega,W(\sqrt{N}\varphi)\mathcal{U}^{*}(t)a_{x}\mathcal{U}(t)\Omega\right\rangle \\
	& \quad+\varphi_{t}(x)\frac{d_{N}}{\sqrt{N}}\left\langle \frac{\left(a^{*}(\varphi)\right)^{N}}{\sqrt{N!}}\Omega,W(\sqrt{N}\varphi)\mathcal{U}^{*}(t)a_{y}^{*}\mathcal{U}(t)\Omega\right\rangle .
	\end{align*}
	Recall the definition of $E_{t}^{1}(J)$ and $E_{t}^{2}(J)$ in Propositions
	\ref{prop:Et1} and \ref{prop:Et2}. For any compact one-particle
	Hermitian operator $J$ on $L^{2}(\mathbb{R}^{3})$, we have 
	\begin{align*}
	& \Tr J(\overline{\gamma}_{N,t}^{(1)}-\left|\phi_{t}\right\rangle \left\langle \phi_{t}\right|)\\
	& =\int\mathrm{d}x\mathrm{d}yJ(x;y)\left(\overline{\gamma}_{N,t}^{(1)}(y;x)-\phi_{t}(y)\overline{\phi_{t}\left(x\right)}\right)\\
	& =\frac{d_{N}}{N}\left\langle \frac{\left(a^{*}(\varphi)\right)^{N}}{\sqrt{N!}}\Omega,W(\sqrt{N}\varphi)\mathcal{U}^{*}(t)d\Gamma(J)\mathcal{U}(t)\Omega\right\rangle \\
	& \quad+\frac{d_{N}}{\sqrt{N}}\left\langle \frac{\left(a^{*}(\varphi)\right)^{N}}{\sqrt{N!}}\Omega,W(\sqrt{N}\varphi)\mathcal{U}^{*}(t)\phi(J\phi_{t})\mathcal{U}(t)\Omega\right\rangle \\
	& =E_{t}^{1}(J)+E_{t}^{2}(J).
	\end{align*}
	Thus, from Propositions \ref{prop:Et1} and \ref{prop:Et2}, we find
	that 
	\[
	\left|\Tr J(\overline{\gamma}_{N,t}^{(1)}-\left|\phi_{t}\right\rangle \left\langle \phi_{t}\right|)\right|\leq C\frac{\left\Vert J\right\Vert }{N^{2-\eta}}e^{Kt}.
	\]
	Since the space of compact operators is the dual to that of the trace
	class operators, and since $\overline{\gamma}_{N,t}^{(1)}$ and $\left|\phi_{t}\right\rangle \left\langle \phi_{t}\right|$
	are Hermitian, 
	\[
	\Tr\left|\overline{\gamma}_{N,t}^{(1)}-\left|\phi_{t}\right\rangle \left\langle \phi_{t}\right|\right|\leq\frac{C}{N^{2-\eta}}e^{Kt}
	\]
	which concludes the proof of Proposition \ref{prop:reged-V}.
\end{proof}

\section{Comparison dynamics\label{sec:comparison}}

As briefly mentioned in Section \ref{sec:introduction}, the key technical
estimate is the upper bound on the fluctuation of the expected number
of particles under the evolution $\mathcal{U}(t;s)$, which is the
following lemma. This section is to provide useful comparison dynamics.
\begin{lem}
	\label{lem:NjU} Suppose that the assumptions in Theorem \ref{thm:main}
	hold. Let $\mathcal{U}\left(t;s\right)$ be the unitary evolution
	defined in \eqref{eq:def_mathcalU}. Then for any $\psi\in\mathcal{F}$
	and $j\in\mathbb{N}$, there exist constants $C\equiv C(j)$ and $K\equiv K(j)$
	such that 
	\[
	\left\langle \mathcal{U}\left(t;s\right)\psi,\mathcal{N}^{j}\mathcal{U}\left(t;s\right)\psi\right\rangle \leq Ce^{Kt}\left\langle \psi,\left(\mathcal{N}+1\right)^{2j+3}\psi\right\rangle .
	\]
\end{lem}

We now begin the proof of Lemma \ref{lem:NjU}. First, we introduce
a truncated time-dependent generator with fixed $M>0$ as follows:
\begin{align*}
& \mathcal{L}_{N}^{(M)}(t)\\
& =\int\mathrm{d}xa_{x}^{*}(-\Delta_{x})a_{x}\\
& \quad+\frac{1}{6}\int\mathrm{d}x\mathrm{d}y\mathrm{d}z\,\overline{V}(x-y,x-z)\Bigg[3|\varphi_{t}(x)|^{2}\left(\varphi_{t}(y)\varphi_{t}(z)a_{y}^{*}a_{z}^{*}+\overline{\varphi_{t}(y)}\overline{\varphi_{t}(z)}a_{z}a_{y}\right)\\
& \qquad\qquad\qquad\qquad\qquad\qquad+3|\varphi_{t}(x)|^{2}\left(\varphi_{t}(y)\overline{\varphi_{t}(z)}a_{y}^{*}a_{z}+\overline{\varphi_{t}(y)}\varphi_{t}(z)a_{z}^{*}a_{y}\right)\Bigg]\\
& \quad+\frac{1}{6\sqrt{N}}\int\mathrm{d}x\mathrm{d}y\mathrm{d}z\,\overline{V}(x-y,x-z)\chi(\mathcal{N}\leq M)\\
&\qquad\qquad\qquad\times\Bigg[\left(\varphi_{t}(x)\varphi_{t}(y)\varphi_{t}(z)a_{x}^{*}a_{y}^{*}a_{z}^{*}+\overline{\varphi_{t}(x)}\overline{\varphi_{t}(y)}\overline{\varphi_{t}(z)}a_{x}a_{y}a_{z}\right)\\
&\qquad\qquad\qquad+3\left(\varphi_{t}(x)\varphi_{t}(y)\overline{\varphi_{t}(z)}a_{x}^{*}a_{y}^{*}a_{z}+\overline{\varphi_{t}(x)}\overline{\varphi_{t}(y)}\varphi_{t}(z)a_{z}^{*}a_{x}a_{y}\right)\Bigg]\\
& \quad+\frac{1}{6N}\int\mathrm{d}x\mathrm{d}y\mathrm{d}z\,\overline{V}(x-y,x-z)\chi(\mathcal{N}\leq M)\Bigg[3\left(\varphi_{t}(x)\varphi_{t}(y)a_{x}^{*}a_{y}^{*}a_{z}^{*}a_{z}+\overline{\varphi_{t}(x)}\overline{\varphi_{t}(y)}a_{x}a_{y}a_{z}^{*}a_{z}\right)\Bigg]\\
& \quad+\frac{1}{6N}\int\mathrm{d}x\mathrm{d}y\mathrm{d}z\,\overline{V}(x-y,x-z)\chi(\mathcal{N}\leq M)\Bigg[6\varphi_{t}(x)\overline{\varphi_{t}(y)}a_{x}^{*}a_{z}^{*}a_{z}a_{y}+3|\varphi_{t}(x)|^{2}a_{y}^{*}a_{z}^{*}a_{z}a_{y}\Bigg]\\
& \quad+\frac{1}{2N\sqrt{N}}\int\mathrm{d}x\mathrm{d}y\mathrm{d}z\,\overline{V}(x-y,x-z)\chi(\mathcal{N}\leq M)a_{x}^{*}a_{y}^{*}\left(\varphi_{t}(z)a_{z}^{*}+\overline{\varphi_{t}(z)}a_{z}\right)a_{y}a_{x}\\
& \quad+\frac{1}{6N^{2}}\int\mathrm{d}x\mathrm{d}y\mathrm{d}z\,\overline{V}(x-y,x-z)a_{x}^{*}a_{y}^{*}a_{z}^{*}a_{z}a_{y}a_{x}.
\end{align*}
We remark that $M$ will be chosen to be $M=N^{1/3}$ later in the
proof of Lemma \ref{lem:NjU}. Define a unitary operator $\mathcal{U}^{(M)}$
by 
\begin{equation}
\mathrm{i}\partial_{t}\mathcal{U}^{(M)}\left(t;s\right)=\mathcal{L}_{N}^{(M)}(t)\mathcal{U}^{(M)}(t;s)\quad\text{and}\quad\mathcal{U}^{(M)}\left(s;s\right)=1.\label{eq:def_mathcalUM}
\end{equation}
We use a three-step strategy.
\vspace{2em}

\noindent \emph{Step 1. Truncation with respect to $\mathcal{N}$
	with $M>0$.} 
\begin{lem}
	\label{lem:truncation} Suppose that the assumptions in Theorem \ref{thm:main}
	hold and let $\mathcal{U}^{(M)}$ be the unitary operator defined
	in \eqref{eq:def_mathcalUM}. Then, there exist constants $C$ and
	$K$ such that, for all $N\in\mathbb{N}$ and $M>0$, $\psi\in\mathcal{F}$,
	and $t,s,\in\mathbb{R}$, 
	\begin{align*}
	&\left\langle \mathcal{U}^{(M)}(t;s)\psi,(\mathcal{N}+1)^{j}\mathcal{U}^{(M)}(t;s)\psi\right\rangle \\
	&\leq\left\langle \psi,(\mathcal{N}+1)^{j}\psi\right\rangle C\exp\left(4^{j}K|t-s|\left(1+\sqrt{\frac{M}{N}}+\frac{M}{N}+\alpha_{N}^{-1}\left(\frac{M}{N}\right)^{3/2}\right)\right).
	\end{align*}
\end{lem}

\begin{proof}
	Following the proof of Lemma 3.5 in \cite{Rodnianski2009}, we get
	\begin{align*}
	& \frac{\mathrm{d}}{\mathrm{d}t}\left\langle \mathcal{U}^{(M)}(t;0)\psi,(\mathcal{N}+1)^{j}\mathcal{U}^{(M)}(t;0)\psi\right\rangle \\
	& =\left\langle \mathcal{U}^{(M)}(t;0)\psi,[\mathrm{i}\mathcal{L}_{N}^{(M)}(t),(\mathcal{N}+1)^{j}]\,\mathcal{U}^{(M)}(t;0)\psi\right\rangle \\
	& =\operatorname{Im}\int\mathrm{d}x\mathrm{d}y\mathrm{d}z\,\overline{V}(x-y,x-z)|\varphi_{t}(x)|^{2}\varphi_{t}(y)\varphi_{t}(z)\left\langle \mathcal{U}^{(M)}(t;0)\psi,[a_{y}^{*}a_{z}^{*},(\mathcal{N}+1)^{j}]\,\mathcal{U}^{(M)}(t;0)\psi\right\rangle \\
	& \quad+\frac{1}{3\sqrt{N}}\operatorname{Im}\int\mathrm{d}x\mathrm{d}y\mathrm{d}z\,\overline{V}(x-y,x-z)\varphi_{t}(x)\varphi_{t}(y)\varphi_{t}(z)\left\langle \mathcal{U}^{(M)}(t;0)\psi,[a_{x}^{*}a_{y}^{*}a_{z}^{*}\chi(\mathcal{N}\leq M),(\mathcal{N}+1)^{j}]\,\mathcal{U}^{(M)}(t;0)\psi\right\rangle \\
	& \quad+\frac{1}{\sqrt{N}}\operatorname{Im}\int\mathrm{d}x\mathrm{d}y\mathrm{d}z\,\overline{V}(x-y,x-z)\varphi_{t}(x)\varphi_{t}(y)\overline{\varphi_{t}(z)}\left\langle \mathcal{U}^{(M)}(t;0)\psi,[a_{x}^{*}a_{y}^{*}a_{z}\chi(\mathcal{N}\leq M),(\mathcal{N}+1)^{j}]\,\mathcal{U}^{(M)}(t;0)\psi\right\rangle \\
	& \quad+\frac{1}{N}\operatorname{Im}\int\mathrm{d}x\mathrm{d}y\mathrm{d}z\,\overline{V}(x-y,x-z)\varphi_{t}(x)\varphi_{t}(y)\left\langle \mathcal{U}^{(M)}(t;0)\psi,[a_{x}^{*}a_{y}^{*}a_{z}^{*}a_{z}\chi(\mathcal{N}\leq M),(\mathcal{N}+1)^{j}]\,\mathcal{U}^{(M)}(t;0)\psi\right\rangle \\
	& \quad+\frac{1}{N\sqrt{N}}\operatorname{Im}\int\mathrm{d}x\mathrm{d}y\mathrm{d}z\,\overline{V}(x-y,x-z)\varphi_{t}(z)\left\langle \mathcal{U}^{(M)}(t;0)\psi,[a_{x}^{*}a_{y}^{*}\chi(\mathcal{N}\leq M)a_{z}^{*}a_{y}a_{x},(\mathcal{N}+1)^{j}]\,\mathcal{U}^{(M)}(t;0)\psi\right\rangle .
	\end{align*}
	Using the pull-through formulae $a_{x}\mathcal{N}=(\mathcal{N}+1)a_{x}$
	and $a_{x}^{*}\mathcal{N}=(\mathcal{N}-1)a_{x}^{*}$, we find
	\[
	[a_{x}^{*},(\mathcal{N}+1)^{j}]=\sum_{k=0}^{j-1}{j \choose k}(-1)^{k}(\mathcal{N}+1)^{k}a_{x}^{*},\quad[a_{x},(\mathcal{N}+1)^{j}]=\sum_{k=0}^{j-1}{j \choose k}(-1)^{k}(\mathcal{N}+1)^{k}a_{x}.
	\]
	As a consequence,
	\begin{align*}
	&[a_{x}^{*}a_{y}^{*},(\mathcal{N}+1)^{j}]\\ & =\sum_{k=0}^{j-1}{j \choose k}(-1)^{k}\left(a_{x}^{*}(\mathcal{N}+1)^{k}a_{y}^{*}+(\mathcal{N}+1)^{k}a_{x}^{*}a_{y}^{*}\right)\\
	& =\sum_{k=0}^{j-1}{j \choose k}(-1)^{k}\left(\mathcal{N}^{k/2}a_{x}^{*}a_{y}^{*}(\mathcal{N}+1)^{k/2}+(\mathcal{N}+1)^{k/2}a_{x}^{*}a_{y}^{*}(\mathcal{N}+3)^{k/2}\right)\\{}
	&[a_{x},(\mathcal{N}+1)^{j}]\\ & =\sum_{k=0}^{j-1}{j \choose k}(-1)^{k}(\mathcal{N}+1)^{k}a_{x}=\sum_{k=0}^{j-1}{j \choose k}(-1)^{k}(\mathcal{N}+1)^{k/2}a_{x}\mathcal{N}^{k/2}.
	\end{align*}
	Moreover,
	\begin{align*}
	&[a_{x}^{*}a_{y}^{*}a_{z}^{*},(\mathcal{N}+1)^{j}]\\
	 & =\sum_{k=0}^{j-1}{j \choose k}(-1)^{k}\Big((\mathcal{N}-1)^{k/2}a_{x}^{*}a_{y}^{*}a_{z}^{*}(\mathcal{N}+2)^{k/2}+\mathcal{N}{}^{k/2}a_{x}^{*}a_{y}^{*}a_{z}^{*}(\mathcal{N}+3)^{k/2}\\
	& \qquad+(\mathcal{N}+1)^{k/2}a_{x}^{*}a_{y}^{*}a_{z}^{*}(\mathcal{N}+4)^{k/2}\Big)\\{}
	&[a_{x}^{*}a_{y}^{*}a_{z},(\mathcal{N}+1)^{j}] \\
	& =\sum_{k=0}^{j-1}{j \choose k}(-1)^{k}\Big((\mathcal{N}-1)^{k/2}a_{x}^{*}a_{y}^{*}a_{z}\mathcal{N}{}^{k/2}+\mathcal{N}{}^{k/2}a_{x}^{*}a_{y}^{*}a_{z}(\mathcal{N}+1)^{k/2}\\
	& \qquad+(\mathcal{N}+1)^{k/2}a_{x}^{*}a_{y}^{*}a_{z}(\mathcal{N}+2)^{k/2}\Big)\\{}
	&[a_{x}^{*}a_{y}^{*}a_{z}^{*}a_{z},(\mathcal{N}+1)^{j}]\\
	& =\sum_{k=0}^{j-1}{j \choose k}(-1)^{k}\Big((\mathcal{N}-2)^{k/2}a_{x}^{*}a_{y}^{*}a_{z}^{*}a_{z}\mathcal{N}{}^{k/2}+(\mathcal{N}-1){}^{k/2}a_{x}^{*}a_{y}^{*}a_{z}^{*}a_{z}(\mathcal{N}+1)^{k/2}\\
	& \qquad+\mathcal{N}{}^{k/2}a_{x}^{*}a_{y}^{*}a_{z}^{*}a_{z}(\mathcal{N}+2)^{k/2}+(\mathcal{N}+1)^{k/2}a_{x}^{*}a_{y}^{*}a_{z}^{*}a_{z}(\mathcal{N}+3)^{k/2}\Big)\\{}
	&[a_{x}^{*}a_{y}^{*}a_{z}^{*}a_{y}a_{x},(\mathcal{N}+1)^{j}]\\
	& =\sum_{k=0}^{j-1}{j \choose k}(-1)^{k}\Big((\mathcal{N}-3)^{k/2}a_{x}^{*}a_{y}^{*}a_{z}^{*}a_{y}a_{x}(\mathcal{N}-2){}^{k/2}+(\mathcal{N}-2){}^{k/2}a_{x}^{*}a_{y}^{*}a_{z}^{*}a_{y}a_{x}(\mathcal{N}-1)^{k/2}\\
	& \qquad+(\mathcal{N}-1){}^{k/2}a_{x}^{*}a_{y}^{*}a_{z}^{*}a_{y}a_{x}\mathcal{N}^{k/2}+\mathcal{N}{}^{k/2}a_{x}^{*}a_{y}^{*}a_{z}^{*}a_{y}a_{x}(\mathcal{N}+1)^{k/2}\\
	& \qquad+(\mathcal{N}+1)^{k/2}a_{x}^{*}a_{y}^{*}a_{z}^{*}a_{y}a_{x}(\mathcal{N}+2)^{k/2}\Big)
	\end{align*}
	Therefore,
	\begin{align}
	& \frac{\mathrm{d}}{\mathrm{d}t}\left\langle \mathcal{U}^{(M)}(t;0)\psi,(\mathcal{N}+1)^{j}\mathcal{U}^{(M)}(t;0)\psi\right\rangle \nonumber \\
	& =\sum_{k=0}^{j-1}{j \choose k}(-1)^{k}\operatorname{Im}\int\mathrm{d}x\mathrm{d}y\mathrm{d}z\,\overline{V}(x-y,x-z)|\varphi_{t}(x)|^{2}\varphi_{t}(y)\varphi_{t}(z)\nonumber \\
	& \times\left\langle \mathcal{U}^{(M)}(t;0)\psi,\left(\mathcal{N}^{k/2}a_{y}^{*}a_{z}^{*}(\mathcal{N}+2)^{k/2}+(\mathcal{N}+1)^{k/2}a_{y}^{*}a_{z}^{*}(\mathcal{N}+3)^{k/2}\right)\mathcal{U}^{(M)}(t;0)\psi\right\rangle \nonumber \\
	& +\frac{1}{3\sqrt{N}}\sum_{k=0}^{j-1}{j \choose k}\operatorname{Im}\int\mathrm{d}x\mathrm{d}y\mathrm{d}z\,\overline{V}(x-y,x-z)\varphi_{t}(x)\varphi_{t}(y)\varphi_{t}(z)\nonumber \\
	& \times\Bigg\langle\mathcal{U}^{(M)}(t;0)\psi,\sum_{k=0}^{j-1}{j \choose k}(-1)^{k}\chi(\mathcal{N}\leq M)\Big((\mathcal{N}-1)^{k/2}a_{x}^{*}a_{y}^{*}a_{z}^{*}(\mathcal{N}+2)^{k/2}\nonumber \\
	& \qquad\qquad+\mathcal{N}{}^{k/2}a_{x}^{*}a_{y}^{*}a_{z}^{*}(\mathcal{N}+3)^{k/2}+(\mathcal{N}+1)^{k/2}a_{x}^{*}a_{y}^{*}a_{z}^{*}(\mathcal{N}+4)^{k/2}\Big)\mathcal{U}^{(M)}(t;0)\psi\Bigg\rangle\nonumber \\
	& +\frac{1}{\sqrt{N}}\sum_{k=0}^{j-1}{j \choose k}\operatorname{Im}\int\mathrm{d}x\mathrm{d}y\mathrm{d}z\,\overline{V}(x-y,x-z)\varphi_{t}(x)\varphi_{t}(y)\varphi_{t}(z)\nonumber \\
	& \qquad\times\Bigg\langle\mathcal{U}^{(M)}(t;0)\psi,\sum_{k=0}^{j-1}{j \choose k}(-1)^{k}\chi(\mathcal{N}\leq M)\Big((\mathcal{N}-1)^{k/2}a_{x}^{*}a_{y}^{*}a_{z}\mathcal{N}{}^{k/2}\nonumber \\
	& \qquad\qquad+\mathcal{N}{}^{k/2}a_{x}^{*}a_{y}^{*}a_{z}(\mathcal{N}+1)^{k/2}+(\mathcal{N}+1)^{k/2}a_{x}^{*}a_{y}^{*}a_{z}(\mathcal{N}+2)^{k/2}\Big)\mathcal{U}^{(M)}(t;0)\psi\Bigg\rangle\nonumber \\
	& +\frac{1}{N}\sum_{k=0}^{j-1}{j \choose k}\operatorname{Im}\int\mathrm{d}x\mathrm{d}y\mathrm{d}z\,\overline{V}(x-y,x-z)\varphi_{t}(x)\varphi_{t}(y)\nonumber \\
	& \qquad\times\Bigg\langle\mathcal{U}^{(M)}(t;0)\psi,\chi(\mathcal{N}\leq M)\Big((\mathcal{N}-2)^{k/2}a_{x}^{*}a_{y}^{*}a_{z}^{*}a_{z}\mathcal{N}{}^{k/2}\nonumber \\
	& \qquad\qquad+(\mathcal{N}-1){}^{k/2}a_{x}^{*}a_{y}^{*}a_{z}^{*}a_{z}(\mathcal{N}+1)^{k/2}+\mathcal{N}{}^{k/2}a_{x}^{*}a_{y}^{*}a_{z}^{*}a_{z}(\mathcal{N}+2)^{k/2}\nonumber\\
	& \qquad\qquad+(\mathcal{N}+1)^{k/2}a_{x}^{*}a_{y}^{*}a_{z}^{*}a_{z}(\mathcal{N}+3)^{k/2}\Big)\mathcal{U}^{(M)}(t;0)\psi\Bigg\rangle\nonumber \\
	& +\frac{1}{N\sqrt{N}}\sum_{k=0}^{j-1}{j \choose k}\operatorname{Im}\int\mathrm{d}x\mathrm{d}y\mathrm{d}z\,\overline{V}(x-y,x-z)\varphi_{t}(z)\nonumber \\
	& \qquad\times\Bigg\langle\mathcal{U}^{(M)}(t;0)\psi,\chi(\mathcal{N}\leq M)\Big((\mathcal{N}-3)^{k/2}a_{x}^{*}a_{y}^{*}a_{z}^{*}a_{y}a_{x}(\mathcal{N}-2){}^{k/2}\nonumber \\
	&\qquad\qquad+(\mathcal{N}-2){}^{k/2}a_{x}^{*}a_{y}^{*}a_{z}^{*}a_{y}a_{x}(\mathcal{N}-1)^{k/2} \nonumber\\
	& \qquad\qquad+(\mathcal{N}-1){}^{k/2}a_{x}^{*}a_{y}^{*}a_{z}^{*}a_{y}a_{x}\mathcal{N}^{k/2}+\mathcal{N}{}^{k/2}a_{x}^{*}a_{y}^{*}a_{z}^{*}a_{y}a_{x}(\mathcal{N}+1)^{k/2}\nonumber \\
	& \qquad\qquad+(\mathcal{N}+1)^{k/2}a_{x}^{*}a_{y}^{*}a_{z}^{*}a_{y}a_{x}(\mathcal{N}+2)^{k/2}\Big)\mathcal{U}^{(M)}(t;0)\psi\Bigg\rangle\nonumber \\
	& =:I_{1}+I_{2}+I_{3}+I_{4}+I_{5}.\label{eq:ddtUM}
	\end{align}
	To control the contribution from the terms in the right-hand side
	of \eqref{eq:ddtUM}, we use the bounds of the form
	\begin{align*}
	I_{1}= & \left|\int\mathrm{d}x\mathrm{d}y\mathrm{d}z\,\overline{V}(x-y,x-z)|\varphi_{t}(x)|^{2}\varphi_{t}(y)\varphi_{t}(z)\langle\mathcal{U}^{(M)}(t;0)\psi,(\mathcal{N}+1)^{\frac{k}{2}}a_{y}^{*}a_{z}^{*}(\mathcal{N}+3)^{\frac{k}{2}}\mathcal{U}^{(M)}(t;0)\psi\rangle\right|\\
	& \quad\leq\int\mathrm{d}x\mathrm{d}y\,\overline{V}(x-y)|\varphi_{t}(x)|^{2}|\varphi_{t}(y)|\|a_{y}(\mathcal{N}+1)^{\frac{k}{2}}\mathcal{U}^{(M)}(t;0)\psi\|\|a^{*}(\overline{v}(x-\cdot)\varphi_{t})(\mathcal{N}+3)^{\frac{k}{2}}\mathcal{U}^{(M)}(t;0)\psi\|\\
	& \quad\leq K\|(\mathcal{N}+3)^{\frac{k+1}{2}}\mathcal{U}^{(M)}(t;0)\psi\|^{2},
	\end{align*}
	\begin{align*}
	&N^{1/2}I_{2}=\\
	& \left|\int\mathrm{d}x\mathrm{d}y\mathrm{d}z\,\overline{V}(x-y,x-z)\varphi_{t}(x)\varphi_{t}(y)\varphi_{t}(z)\langle\mathcal{U}^{(M)}(t;0)\psi,(\mathcal{N}+1)^{\frac{k}{2}}a_{x}^{*}a_{y}^{*}a_{z}^{*}(\mathcal{N}+3)^{\frac{k}{2}}\mathcal{U}^{(M)}(t;0)\psi\rangle\right|\\
	& \quad\leq\int\mathrm{d}x\mathrm{d}y\,\overline{v}(x-y)|\varphi_{t}(x)||\varphi_{t}(y)|\|a_{x}a_{y}\chi(\mathcal{N}\leq M)(\mathcal{N}+1)^{\frac{k}{2}}\mathcal{U}^{(M)}(t;0)\psi\|\\
	& \qquad\qquad\qquad\times\|a^{*}(\overline{v}(x-\cdot)\varphi_{t})\chi(\mathcal{N}\leq M)(\mathcal{N}+3)^{\frac{k}{2}}\mathcal{U}^{(M)}(t;0)\psi\|\\
	& \quad\leq KM^{1/2}\|(\mathcal{N}+3)^{\frac{k+1}{2}}\mathcal{U}^{(M)}(t;0)\psi\|^{2},
	\end{align*}
	\begin{align*}
	&N^{1/2}I_{3}= \\
	& \left|\int\mathrm{d}x\mathrm{d}y\mathrm{d}z\,\overline{V}(x-y,x-z)\varphi_{t}(x)\varphi_{t}(y)\varphi_{t}(z)\langle\mathcal{U}^{(M)}(t;0)\psi,(\mathcal{N}+1)^{\frac{k}{2}}a_{x}^{*}a_{y}^{*}a_{z}(\mathcal{N}+3)^{\frac{k}{2}}\mathcal{U}^{(M)}(t;0)\psi\rangle\right|\\
	& \quad\leq\int\mathrm{d}x\mathrm{d}y\,\overline{v}(x-y)|\varphi_{t}(x)||\varphi_{t}(y)|\|a_{x}a_{y}\chi(\mathcal{N}\leq M)(\mathcal{N}+1)^{\frac{k}{2}}\mathcal{U}^{(M)}(t;0)\psi\|\\
	& \qquad\qquad\qquad\times\|a(v(x-\cdot)\varphi_{t})\chi(\mathcal{N}\leq M)(\mathcal{N}+3)^{\frac{k}{2}}\mathcal{U}^{(M)}(t;0)\psi\|\\
	& \quad\leq KM^{1/2}\|(\mathcal{N}+3)^{\frac{k+1}{2}}\mathcal{U}^{(M)}(t;0)\psi\|^{2},
	\end{align*}
	and
	\begin{align*}
	&NI_{4}=\\
	&\left|\int\mathrm{d}x\mathrm{d}y\mathrm{d}z\,\overline{V}(x-y,x-z)\varphi_{t}(x)\varphi_{t}(y)\langle\mathcal{U}^{(M)}(t;0)\psi,(\mathcal{N}+1)^{\frac{k}{2}}a_{x}^{*}a_{y}^{*}a_{z}^{*}a_{z}(\mathcal{N}+3)^{\frac{k}{2}}\mathcal{U}^{(M)}(t;0)\psi\rangle\right|\\
	& \quad\leq\int\mathrm{d}x\mathrm{d}y\mathrm{d}z\,\overline{V}(x-y,x-z)|\varphi_{t}(x)||\varphi_{t}(y)|\\
	&\qquad\qquad\qquad\times\|a_{z}a_{y}a_{x}\chi(\mathcal{N}\leq M)(\mathcal{N}+1)^{\frac{k}{2}}\mathcal{U}^{(M)}(t;0)\psi\|\|a_{z}(\mathcal{N}+3)^{\frac{k}{2}}\mathcal{U}^{(M)}(t;0)\psi\|\\
	& \quad\leq KM\|(\mathcal{N}+3)^{\frac{k+1}{2}}\mathcal{U}^{(M)}(t;0)\psi\|^{2}.
	\end{align*}
	On the other hand, to control contribution arising from $I_{5}$ in
	the right-hand side of \eqref{eq:ddtUM}, we use that 
	\begin{align*}
	&N^{3/2}I_{5}= \\& \left|\int\mathrm{d}x\mathrm{d}y\,\overline{V}(x-y,x-z)\varphi_{t}(z)\langle\mathcal{U}^{(M)}(t;0)\psi,(\mathcal{N}+1)^{k/2}a_{x}^{*}a_{y}^{*}a_{z}^{*}a_{y}a_{x}\chi(\mathcal{N}\leq M)(\mathcal{N}+2)^{k/2}\mathcal{U}^{(M)}(t;0)\psi\rangle\right|\\
	& \quad\leq\int\mathrm{d}x\,\alpha_{N}^{-1}\left\Vert a_{x}a_{y}a_{z}\chi(\mathcal{N}\leq M)(\mathcal{N}+1)^{\frac{k}{2}}\mathcal{U}^{(M)}(t;0)\psi\right\Vert\\
	&\qquad\qquad\times \left\Vert a(v(x-\cdot)\varphi_{t})\chi(\mathcal{N}\leq M)\right\Vert \left\Vert a_{x}\mathcal{N}^{\frac{k}{2}}\mathcal{U}^{(M)}(t;0)\psi\right\Vert \\
	& \quad\leq K\alpha_{N}^{-1}M^{3/2}\sup_{x}\|v(x-\cdot)\varphi_{t}\|\|(\mathcal{N}+1)^{\frac{k+1}{2}}\mathcal{U}^{(M)}(t;0)\psi\|^{2}.
	\end{align*}
	This implies 
	\begin{align*}
	& \frac{\mathrm{d}}{\mathrm{d}t}\left\langle \mathcal{U}^{(M)}(t;0)\psi,(\mathcal{N}+1)^{j}\mathcal{U}^{(M)}(t;0)\psi\right\rangle \\
	& \quad\leq K\left(1+\sqrt{\frac{M}{N}}+\frac{M}{N}+\alpha_{N}^{-1}\left(\frac{M}{N}\right)^{3/2}\right)\sum_{k=0}^{j}{j \choose k}\left\langle \mathcal{U}^{(M)}(t;0)\psi,(\mathcal{N}+3)^{j}\mathcal{U}^{(M)}(t;0)\psi\right\rangle \\
	& \quad\leq4^{j}K\left(1+\sqrt{\frac{M}{N}}+\frac{M}{N}+\alpha_{N}^{-1}\left(\frac{M}{N}\right)^{3/2}\right)\left\langle \mathcal{U}^{(M)}(t;0)\psi,(\mathcal{N}+1)^{j}\mathcal{U}^{(M)}(t;0)\psi\right\rangle .
	\end{align*}
	Applying the Grönwall Lemma with Lemma \ref{lem:V2phiphi}, we get
	the desired result.
\end{proof}
\noindent \emph{Step 2: Weak bounds on the $\mathcal{U}$ dynamics.} 
\begin{lem}
	\label{lem:UNU} For arbitrary $t,s\in\mathbb{R}$ and $\psi\in\mathcal{F}$,
	we have 
	\[
	\left\langle \psi,\mathcal{U}(t;s)\mathcal{N}\mathcal{U}(t;s)^{*}\psi\right\rangle \leq6\left\langle \psi,(\mathcal{N}+N+1)\psi\right\rangle .
	\]
	Moreover, for every $\ell\in\mathbb{N}$, there exists a constant
	$C(\ell)$ such that 
	\[
	\left\langle \psi,\mathcal{U}(t;s)\mathcal{N}^{2\ell}\mathcal{U}(t;s)^{*}\psi\right\rangle \leq C(\ell)\left\langle \psi,(\mathcal{N}+N)^{2\ell}\psi\right\rangle ,
	\]
	\[
	\left\langle \psi,\mathcal{U}(t;s)\mathcal{N}^{2\ell+1}\mathcal{U}(t;s)^{*}\psi\right\rangle \leq C(\ell)\left\langle \psi,(\mathcal{N}+N)^{2\ell+1}(\mathcal{N}+1)\psi\right\rangle 
	\]
	for all $t,s\in\mathbb{R}$ and $\psi\in\mathcal{F}$. 
\end{lem}

\begin{proof}
	We may follow the proof of Lemma 3.6 in \cite{Rodnianski2009} without
	any change. 
\end{proof}
\noindent \emph{Step 3: Comparison between the $\mathcal{U}$ and
	$\mathcal{U}^{(M)}$ dynamics.} 
\begin{lem}
	\label{lem:comparison} Suppose that the assumptions in Theorem \ref{thm:main}
	hold. Then, for every $j\in\mathbb{N}$, there exist constants $C\equiv C(j)$
	and $K\equiv K(j)$ such that 
	\begin{align*}
	& \left|\left\langle \mathcal{U}(t;s)\psi,\mathcal{N}^{j}\left(\mathcal{U}(t;s)-\mathcal{U}^{(M)}(t;s)\right)\psi\right\rangle \right|\\
	& \quad\leq C(j)\frac{\left((N/M)^{j}+(N/M)^{j-1/2}+(N/M)^{j-1}\right)\|(\mathcal{N}+1)^{j+3/2}\psi\|^{2}}{\left(1+\sqrt{\frac{M}{N}}+\frac{M}{N}+\alpha_{N}^{-1}\left(\frac{M}{N}\right)^{3/2}\right)}\\
	&\qquad\times\exp\left(K(j)|t-s|\left(1+\sqrt{\frac{M}{N}}+\frac{M}{N}+\alpha_{N}^{-1}\left(\frac{M}{N}\right)^{3/2}\right)\right)
	\end{align*}
	and 
	\begin{align*}
	& \left|\left\langle \mathcal{U}^{(M)}(t;s)\psi,\mathcal{N}^{j}\left(\mathcal{U}(t;s)-\mathcal{U}^{(M)}(t;s)\right)\psi\right\rangle \right|\\
	& \quad\leq C(j)\frac{\left((1/M)^{j}+(1/M)^{j-1/2}+(1/M)^{j-1}\right)\|(\mathcal{N}+1)^{j+3/2}\psi\|^{2}}{\left(1+\sqrt{\frac{M}{N}}+\frac{M}{N}+\alpha_{N}^{-1}\left(\frac{M}{N}\right)^{3/2}\right)}\\
	&\qquad\times\exp\left(K(j)|t-s|\left(1+\sqrt{\frac{M}{N}}+\frac{M}{N}+\alpha_{N}^{-1}\left(\frac{M}{N}\right)^{3/2}\right)\right).
	\end{align*}
\end{lem}

\begin{proof}
	To simplify the notation we consider the case $s=0$ and $t>0$ only;
	other cases can be treated in a similar manner. To prove the first
	inequality of the lemma, since $\alpha_{N}=N^{-\eta},$we expand the
	difference of the two evolution as follows: 
	\begin{align*}
	& \left\langle \mathcal{U}(t;0)\psi,\mathcal{N}^{j}\left(\mathcal{U}(t;0)-\mathcal{U}^{(M)}(t;0)\right)\psi\right\rangle =\left\langle \mathcal{U}(t;0)\psi,\mathcal{N}^{j}\mathcal{U}(t;0)\left(1-\mathcal{U}(t;0)^{*}\mathcal{U}^{(M)}(t;0)\right)\psi\right\rangle \\
	& =-\mathrm{i}\int_{0}^{t}\mathrm{d}s\,\left\langle \mathcal{U}(t;0)\psi,\mathcal{N}^{j}\mathcal{U}(t;s)\left(\mathcal{L}_{N}(s)-\mathcal{L}_{N}^{(M)}(s)\right)\mathcal{U}^{(M)}(s;0)\psi\right\rangle \\
	& =J_{1}+J_{2}+J_{3}+J_{4}+J_{5}
	\end{align*}
	where
	\begin{align*}
	&J_{1}:=\\ & -\frac{\mathrm{i}}{3\sqrt{N}}\int_{0}^{t}\mathrm{d}s\int\mathrm{d}x\mathrm{d}y\mathrm{d}z\,\overline{V}(x-y,x-z)\\
	& \qquad\times\Bigg\langle\mathcal{U}(t;0)\psi,\mathcal{N}^{j}\mathcal{U}(t;s)\chi(\mathcal{N}>M)\left(\varphi_{t}(x)\varphi_{t}(y)\varphi_{t}(z)a_{x}^{*}a_{y}^{*}a_{z}^{*}+\overline{\varphi_{t}(x)}\overline{\varphi_{t}(y)}\overline{\varphi_{t}(z)}a_{x}a_{y}a_{z}\right)\mathcal{U}^{(M)}(s;0)\psi\Bigg\rangle,\\
	&J_{2}:=\\ & -\frac{\mathrm{i}}{\sqrt{N}}\int_{0}^{t}\mathrm{d}s\int\mathrm{d}x\mathrm{d}y\mathrm{d}z\,\overline{V}(x-y,x-z)\\
	& \qquad\times\Bigg\langle\mathcal{U}(t;0)\psi,\mathcal{N}^{j}\mathcal{U}(t;s)\chi(\mathcal{N}>M)\left(\varphi_{t}(x)\varphi_{t}(y)\overline{\varphi_{t}(z)}a_{x}^{*}a_{y}^{*}a_{z}+\overline{\varphi_{t}(x)}\overline{\varphi_{t}(y)}\varphi_{t}(z)a_{z}^{*}a_{x}a_{y}\right)\Bigg]\mathcal{U}^{(M)}(s;0)\psi\Bigg\rangle,\\
	&J_{3}:=\\ &-\frac{\mathrm{i}}{N}\int_{0}^{t}\mathrm{d}s\int\mathrm{d}x\mathrm{d}y\mathrm{d}z\,\overline{V}(x-y,x-z)\\
	& \qquad\times\left\langle \mathcal{U}(t;0)\mathcal{N}^{j}\mathcal{U}(t;0)\psi,\chi(\mathcal{N}>M)\left(\varphi_{t}(x)\varphi_{t}(y)a_{x}^{*}a_{y}^{*}a_{z}^{*}a_{z}+\overline{\varphi_{t}(x)}\overline{\varphi_{t}(y)}a_{x}a_{y}a_{z}^{*}a_{z}\right)\,\mathcal{U}^{(M)}(s;0)\psi\right\rangle ,\\
	&J_{4}:=\\ &-\frac{\mathrm{i}}{6N}\int_{0}^{t}\mathrm{d}s\int\mathrm{d}x\mathrm{d}y\mathrm{d}z\,\overline{V}(x-y,x-z)\\
	& \qquad\times\left\langle \mathcal{U}(t;0)\mathcal{N}^{j}\mathcal{U}(t;s)\psi,\chi(\mathcal{N}>M)\left(6\varphi_{t}(x)\overline{\varphi_{t}(y)}a_{x}^{*}a_{z}^{*}a_{z}a_{y}+3|\varphi_{t}(x)|^{2}a_{y}^{*}a_{z}^{*}a_{z}a_{y}\right)\,\mathcal{U}^{(M)}(s;0)\psi\right\rangle ,\\
	\intertext{and}J_{5} & :=-\frac{\mathrm{i}}{N\sqrt{N}}\int_{0}^{t}\mathrm{d}s\int\mathrm{d}x\mathrm{d}y\mathrm{d}z\,\overline{V}(x-y,x-z)\varphi_{t}(z)\\
	& \qquad\times\left\langle \mathcal{U}(t;0)\mathcal{N}^{j}\mathcal{U}(t;s)\psi,a_{x}^{*}a_{y}^{*}\left(a_{z}\chi(\mathcal{N}>M)+\chi(\mathcal{N}>M)a_{z}^{*}\right)a_{y}a_{x}\,\mathcal{U}^{(M)}(s;0)\psi\right\rangle .
	\end{align*}
	Note that $\chi(\mathcal{N}>M)\leq(\mathcal{N}/M)^{k}$for any $k\geq1$
	and also note from Lemma \ref{lem:UNU} that
	\begin{align}
	\|\mathcal{N}\mathcal{U}(t;s)^{*}\mathcal{N}^{j}\mathcal{U}(t;0)\psi\|^{2} & \leq6\langle\mathcal{N}^{j}\mathcal{U}(t;0)\psi,(\mathcal{N}+N+1)^{2}\mathcal{N}^{j}\mathcal{U}(t;0)\psi\rangle\\
	& \leq C(j)\langle\psi,(\mathcal{N}+N)^{2j+2}(\mathcal{N}+1)\psi\rangle\leq C(j)N^{2j+2}\langle\psi,(\mathcal{N}+1)^{2j+3}\psi\rangle.\nonumber 
	\end{align}
	Then
	\begin{align*}
	|J_{2}| & \leq\frac{C}{\sqrt{N}}\int_{0}^{t}\mathrm{d}s\int\mathrm{d}x\,\varphi_{t}(x)\|a_{x}\mathcal{U}(t;s)^{*}\mathcal{N}^{j}\mathcal{U}(t;0)\psi\|\\
	& \qquad\times\|a(v(x-\cdot)\varphi_{t})a(v(x-\cdot)\varphi_{t})\chi(\mathcal{N}>M+1)\mathcal{U}^{(M)}(s;0)\psi\|\\
	& \leq\frac{C}{\sqrt{N}}\sup_{x}\|v(x-\cdot)\varphi_{t}\|^{2}\int_{0}^{t}\mathrm{d}s\int\mathrm{d}x\,\varphi_{t}(x)\|a_{x}\mathcal{U}(t;s)^{*}\mathcal{N}^{j}\mathcal{U}(t;0)\psi\|\,\|\chi(\mathcal{N}>M+1)\mathcal{U}^{(M)}(s;0)\psi\|\\
	& \leq\frac{C}{\sqrt{N}}\sup_{x}\|v(x-\cdot)\varphi_{t}\|^{2}\int_{0}^{t}\mathrm{d}s\|\mathcal{N}^{1/2}\mathcal{U}(t;s)^{*}\mathcal{N}^{j}\mathcal{U}(t;0)\psi\|\,\|\chi(\mathcal{N}>M+1)\mathcal{U}^{(M)}(s;0)\psi\|\\
	& \leq C(j)N^{j}\|(\mathcal{N}+1)^{j}\psi\|\int_{0}^{t}\mathrm{d}s\left\langle \mathcal{U}^{(M)}(s;0)\psi,\frac{\mathcal{N}^{2j}}{M^{2j}}\mathcal{U}^{(M)}(s;0)\psi\right\rangle ^{1/2}\\
	& \leq C(j)(N/M)^{j}\|(\mathcal{N}+1)^{j}\psi\|\int_{0}^{t}\mathrm{d}s\left\langle \mathcal{U}^{(M)}(s;0)\psi,\mathcal{N}^{2j}\mathcal{U}^{(M)}(s;0)\psi\right\rangle ^{1/2}.
	\end{align*}
	For $J_{1}$, using similar approach, we can obtain the same bound.
	Then we consider $J_{3}$.
	\begin{align*}
	|J_{3}| & \leq\frac{C}{N}\int_{0}^{t}\mathrm{d}s\int\mathrm{d}z\,\|a_{z}\mathcal{U}(t;s)^{*}\mathcal{N}^{j}\mathcal{U}(t;0)\psi\|\\
	& \qquad\times\|a(v(x-\cdot)\varphi_{t})a(v(x-\cdot)\varphi_{t})a_{z}\chi(\mathcal{N}>M+1)\mathcal{U}^{(M)}(s;0)\psi\|\\
	& \leq\frac{C}{N}\sup_{x}\|v(x-\cdot)\varphi_{t}\|^{2}\int_{0}^{t}\mathrm{d}s\int\mathrm{d}z\,\|a_{z}\mathcal{U}(t;s)^{*}\mathcal{N}^{j}\mathcal{U}(t;0)\psi\|\,\|\chi(\mathcal{N}>M+1)a_{z}\mathcal{U}^{(M)}(s;0)\psi\|\\
	& \leq\frac{C}{N}\sup_{x}\|v(x-\cdot)\varphi_{t}\|^{2}\int_{0}^{t}\mathrm{d}s\|\mathcal{N}^{1/2}\mathcal{U}(t;s)^{*}\mathcal{N}^{j}\mathcal{U}(t;0)\psi\|\,\|\chi(\mathcal{N}>M+1)\mathcal{N}^{1/2}\mathcal{U}^{(M)}(s;0)\psi\|\\
	& \leq C(j)N^{j-1/2}\|(\mathcal{N}+1)^{j+1}\psi\|\int_{0}^{t}\mathrm{d}s\left\langle \mathcal{U}^{(M)}(s;0)\psi,\frac{\mathcal{N}^{2j}}{M^{2j-1}}\mathcal{U}^{(M)}(s;0)\psi\right\rangle ^{1/2}\\
	& \leq C(j)(N/M)^{j-1/2}\|(\mathcal{N}+1)^{j+1}\psi\|\int_{0}^{t}\mathrm{d}s\left\langle \mathcal{U}^{(M)}(s;0)\psi,\mathcal{N}^{2j}\mathcal{U}^{(M)}(s;0)\psi\right\rangle ^{1/2}.
	\end{align*}
	Also for $J_{4},$we con obtain the same bound. Thus
	\begin{align*}
	|J_{5}| & \leq\frac{C}{N\sqrt{N}}\int_{0}^{t}\mathrm{d}s\int\mathrm{d}y\mathrm{d}z\,\|a_{y}a_{z}\mathcal{U}(t;s)^{*}\mathcal{N}^{j}\mathcal{U}(t;0)\psi\|\\
	& \qquad\times\|a(v(z-\cdot)\varphi_{t})a_{y}a_{z}\chi(\mathcal{N}>M+1)\mathcal{U}^{(M)}(s;0)\psi\|\\
	& \leq\frac{C}{N\sqrt{N}}\sup_{x}\|v(x-\cdot)\varphi_{t}\|\int_{0}^{t}\mathrm{d}s\int\mathrm{d}y\mathrm{d}z\,\|a_{y}a_{z}\mathcal{U}(t;s)^{*}\mathcal{N}^{j}\mathcal{U}(t;0)\psi\|\,\|a_{y}a_{z}\chi(\mathcal{N}>M+1)\mathcal{U}^{(M)}(s;0)\psi\|\\
	& \leq\frac{C}{N\sqrt{N}}\sup_{x}\|v(x-\cdot)\varphi_{t}\|\int_{0}^{t}\mathrm{d}s\,\|\mathcal{N}\mathcal{U}(t;s)^{*}\mathcal{N}^{j}\mathcal{U}(t;0)\psi\|\,\|\mathcal{N}\chi(\mathcal{N}>M+1)\mathcal{U}^{(M)}(s;0)\psi\|\\
	& \leq C(j)N^{j-1}\|(\mathcal{N}+1)^{j+3/2}\psi\|\int_{0}^{t}\mathrm{d}s\left\langle \mathcal{U}^{(M)}(s;0)\psi,\frac{\mathcal{N}^{2j}}{M^{2j-2}}\mathcal{U}^{(M)}(s;0)\psi\right\rangle ^{1/2}\\
	& \leq C(j)(N/M)^{j-1}\|(\mathcal{N}+1)^{j+3/2}\psi\|\int_{0}^{t}\mathrm{d}s\left\langle \mathcal{U}^{(M)}(s;0)\psi,\mathcal{N}^{2j}\mathcal{U}^{(M)}(s;0)\psi\right\rangle ^{1/2}.
	\end{align*}
	We thus conclude that 
	\begin{align}
	& \left|\left\langle \mathcal{U}(t;0)\psi,\mathcal{N}^{j}\left(\mathcal{U}(t;0)-\mathcal{U}^{(M)}(t;0)\right)\psi\right\rangle \right|\nonumber \\
	& \leq C(j)\frac{\left((N/M)^{j}+(N/M)^{j-1/2}+(N/M)^{j-1}\right)\|(\mathcal{N}+1)^{j+3/2}\psi\|^{2}}{\left(1+\sqrt{\frac{M}{N}}+\frac{M}{N}+\alpha_{N}^{-1}\left(\frac{M}{N}\right)^{3/2}\right)}\\
	&\qquad\times\exp\left(K(j)|t-s|\left(1+\sqrt{\frac{M}{N}}+\frac{M}{N}+\alpha_{N}^{-1}\left(\frac{M}{N}\right)^{3/2}\right)\right)\label{eq:U-UM}
	\end{align}
	The proof of the second part of the lemma is similar and we omit it.
\end{proof}
We now prove Lemma \ref{lem:NjU} by combining the three steps above. 
\begin{proof}[Proof of Lemma \ref{lem:NjU}]
	From Lemmas \ref{lem:truncation}, \ref{lem:UNU}, and \ref{lem:comparison}
	with the choice $\alpha_{N}=N^{-\eta}$ and $M=N^{1-(2\eta/3)}$,
	Since $N>1$, we have that
	\begin{align*}
	& \left\langle \mathcal{U}\left(t;s\right)\psi,\mathcal{N}^{j}\mathcal{U}\left(t;s\right)\psi\right\rangle \\
	& =\left\langle \mathcal{U}\left(t;s\right)\psi,\mathcal{N}^{j}(\mathcal{U}-\mathcal{U}^{(M)})\left(t;s\right)\psi\right\rangle +\left\langle (\mathcal{U}-\mathcal{U}^{(M)})\left(t;s\right)\psi,\mathcal{N}^{j}\mathcal{U}^{(M)}\left(t;s\right)\psi\right\rangle \\
	& \qquad+\left\langle \mathcal{U}^{(M)}\left(t;s\right)\psi,\mathcal{N}^{j}\mathcal{U}^{(M)}\left(t;s\right)\psi\right\rangle \\
	& \leq Ce^{K|t-s|}\left\langle \psi,\left(\mathcal{N}+1\right)^{2j+3}\psi\right\rangle .
	\end{align*}
\end{proof}
Recall the definition of $\mathcal{\widetilde{\mathcal{U}}}\left(t;s\right)$
in \eqref{eq:def_mathcaltildeU}. In the next lemma, we prove an estimate
similar to Lemma \ref{lem:NjU} for the evolution with respect to
$\mathcal{\widetilde{\mathcal{U}}}$.
\begin{lem}
	\label{lem:tildeNj}Suppose that the assumptions in Theorem \ref{thm:main}
	hold. We consider another evolution
	\[
	\mathrm{i}\partial_{t}\widetilde{\mathcal{U}}=(\mathcal{L}_{2}+\mathcal{L}_{4}^{r}+\mathcal{L}_{6})\widetilde{\mathcal{U}}.
	\]
	Then
	\begin{equation}
	\left\langle \widetilde{\psi},(\mathcal{N}+1)^{j}\widetilde{\psi}\right\rangle \leq Ce^{Kt}.\label{eq:NUtilde}
	\end{equation}
\end{lem}

Let $\widetilde{\mathcal{L}}=\mathcal{L}_{2}+\mathcal{L}_{4}^{r}+\mathcal{L}_{6}$
and define the unitary operator $\mathcal{\widetilde{\mathcal{U}}}\left(t;s\right)$
by 
\begin{equation}
\mathrm{i}\partial_{t}\widetilde{\mathcal{U}}\left(t;s\right)=\widetilde{\mathcal{L}}\left(t\right)\widetilde{\mathcal{U}}\left(t;s\right)\quad\text{and }\quad\widetilde{\mathcal{U}}\left(s;s\right)=1.\label{eq:def_mathcaltildeU}
\end{equation}
Since $\widetilde{\mathcal{L}}$ does not change the parity of the
number of particles, 
\begin{equation}
\left\langle \Omega,\widetilde{\mathcal{U}}^{*}\left(t;0\right)a_{y}\widetilde{\mathcal{U}}\left(t;0\right)\Omega\right\rangle =\left\langle \Omega,\widetilde{\mathcal{U}}^{*}\left(t;0\right)a_{x}^{*}\widetilde{\mathcal{U}}\left(t;0\right)\Omega\right\rangle =0.\label{eq:Parity_Consevation}
\end{equation}

\begin{proof}
	For \eqref{eq:NUtilde}, we derivate this with respect to time. Then
	\begin{align*}
	& \frac{\mathrm{d}}{\mathrm{d}t}\left\langle \widetilde{\psi},(\mathcal{N}+1)^{j}\widetilde{\psi}\right\rangle =\frac{\mathrm{d}}{\mathrm{d}t}\left\langle \widetilde{\psi},(\mathcal{N}+1)^{j}\widetilde{\psi}\right\rangle \\
	& =\left\langle \widetilde{\psi},[\mathrm{i}(\mathcal{L}_{2}+\mathcal{L}_{4}^{r}+\mathcal{L}_{6}),(\mathcal{N}+1)^{j}]\widetilde{\psi}\right\rangle \\
	& =2\operatorname{Im}\int\mathrm{d}x\mathrm{d}y\mathrm{d}z\,V_{N}(x-y,x-z)|\varphi_{t}(x)|^{2}\varphi_{t}(y)\varphi_{t}(z)\left\langle \widetilde{\psi},[a_{y}^{*}a_{z}^{*},(\mathcal{N}+1)^{j}]\widetilde{\psi}\right\rangle \\
	& =2\operatorname{Im}\int\mathrm{d}x\mathrm{d}y\mathrm{d}z\,V_{N}(x-y,x-z)|\varphi_{t}(x)|^{2}\varphi_{t}(y)\varphi_{t}(z)\\
	&\qquad\times\left\langle (\mathcal{N}+3)^{\frac{j}{2}-1}a_{y}a_{z}\widetilde{\psi},(\mathcal{N}+3)^{1-\frac{j}{2}}\left((\mathcal{N}+1)^{j}-(\mathcal{N}+3)^{j}\right)\widetilde{\psi}\right\rangle 
	\end{align*}
	Then
	\begin{align*}
	& \left|\frac{\mathrm{d}}{\mathrm{d}t}\left\langle \widetilde{\psi},(\mathcal{N}+1)^{j}\widetilde{\psi}\right\rangle \right|\\
	& \leq C\int\mathrm{d}x\mathrm{d}y\mathrm{d}z\,V_{N}(x-y,x-z)|\varphi_{t}(x)|^{2}|\varphi_{t}(y)||\varphi_{t}(z)|\\
	&\qquad\times\left\Vert (\mathcal{N}+3)^{\frac{j}{2}-1}a_{y}a_{z}\widetilde{\psi}\right\Vert \left\Vert (\mathcal{N}+1)^{\frac{j}{2}}\widetilde{\psi}\right\Vert \\
	& \leq C\left(\int\mathrm{d}x\mathrm{d}y\mathrm{d}z\,|V_{N}(x-y,x-z)|^{2}|\varphi_{t}(x)|^{2}|\varphi_{t}(y)|^{2}|\varphi_{t}(z)|^{2}\right)^{1/2}\\
	& \qquad\times\left(\int\mathrm{d}x\mathrm{d}y\mathrm{d}z\,|\varphi_{t}(x)|^{2}\left\Vert (\mathcal{N}+3)^{\frac{j}{2}-1}a_{y}a_{z}\widetilde{\psi}\right\Vert ^{2}\right)^{1/2}\left\Vert (\mathcal{N}+1)^{\frac{j}{2}}\widetilde{\psi}\right\Vert \\
	& \leq C\left\Vert (\mathcal{N}+1)^{\frac{j}{2}}\widetilde{\psi}\right\Vert ^{2}.
	\end{align*}
	Hence by Grönwall's inequality, we obtain \eqref{eq:NUtilde}.
\end{proof}
\begin{lem}
	\label{lem:L3}Suppose that the assumptions in Theorem \ref{thm:main}
	hold. Then, for any for any $\psi\in\mathcal{F}$ and $j\in\mathbb{N}$,
	there exist a constant $C\equiv C(j)$ such that
	\[
	\left\Vert (\mathcal{N}+1)^{j/2}\mathcal{L}_{3}\psi\right\Vert \leq\frac{Ce^{Kt}}{\sqrt{N}}\left\Vert (\mathcal{N}+1)^{(j+3)}\psi\right\Vert .
	\]
\end{lem}

\begin{proof}
	Let 
	\begin{align*}
	A_{3} & =\int\mathrm{d}x\mathrm{d}y\mathrm{d}z\,V(x-y,x-z)\varphi_{t}(x)\varphi_{t}(y)\varphi_{t}(z)a_{x}^{*}a_{y}^{*}a_{z}^{*}
	\end{align*}
	and
	\[
	B_{3}=\int\mathrm{d}x\mathrm{d}y\mathrm{d}z\,V(x-y,x-z)\varphi_{t}(x)\varphi_{t}(y)\overline{\varphi_{t}(z)}a_{x}^{*}a_{y}^{*}a_{z}
	\]
	Then
	\[
	\mathcal{L}_{3}=\frac{1}{6\sqrt{N}}\left((A_{3}+A_{3}^{*})+3(B_{3}+B_{3}^{*})\right).
	\]
	Take any $\xi\in\mathcal{F}$. Then
	\begin{align*}
	& \left\langle \xi,(\mathcal{N}+1)^{j/2}A_{3}^{*}\psi\right\rangle \\
	& =\int\mathrm{d}x\mathrm{d}y\mathrm{d}z\,\overline{V}(x-y,x-z)\varphi_{t}(x)\varphi_{t}(y)\varphi_{t}(z)\left\langle \xi,(\mathcal{N}+1)^{j/2}a_{z}a_{y}a_{x}\psi\right\rangle \\
	& \leq\left(\int\mathrm{d}x\mathrm{d}y\mathrm{d}z\,|\overline{V}(x-y,x-z)|^{2}|\varphi_{t}(x)|^{2}|\varphi_{t}(y)|^{2}|\varphi_{t}(z)|^{2}\left\Vert (\mathcal{N}+1)^{-1/2}\xi\right\Vert ^{2}\right)^{1/2}\\
	&\qquad\times\left(\int\mathrm{d}x\mathrm{d}y\mathrm{d}z\,\left\Vert a_{z}a_{y}a_{x}(\mathcal{N}+1)^{j/2}\psi\right\Vert ^{2}\right)^{1/2}\\
	& \leq C\|\varphi_{t}\|_{H^{1}(\mathbb{R}^{3})}^{2}\|\xi\|\left\Vert (\mathcal{N}+1)^{(j+3)/2}\psi\right\Vert .
	\end{align*}
	Similarly, this holds for $A_{3}$, $B_{3}$, and $B_{3}^{*}$. It
	leads us the lemma.
\end{proof}
\begin{lem}
	\label{lem:L4c}Suppose that the assumptions in Theorem \ref{thm:main}
	hold. Then, for any for any $\psi\in\mathcal{F}$ and $j\in\mathbb{N}$,
	there exist a constant $C\equiv C(j)$ such that
	\[
	\left\Vert (\mathcal{N}+1)^{j/2}\mathcal{L}_{4}^{c}\psi\right\Vert \leq\frac{Ce^{Kt}}{N}\left\Vert (\mathcal{N}+1)^{(j+4)/2}\psi\right\Vert .
	\]
\end{lem}

\begin{proof}
	Let 
	\begin{align*}
	A_{4} & =\int\mathrm{d}x\mathrm{d}y\mathrm{d}z\,\overline{V}(x-y,x-z)\varphi_{t}(x)\varphi_{t}(y)a_{x}^{*}a_{y}^{*}a_{z}^{*}a_{z}.
	\end{align*}
	Then
	\[
	\mathcal{L}_{4}^{c}=\frac{1}{2N}(A_{4}+A_{4}^{*}).
	\]
	Take any $\xi\in\mathcal{F}$. Then
	\begin{align*}
	& \left\langle \xi,(\mathcal{N}+1)^{j/2}A_{4}^{*}\psi\right\rangle \\
	& =\int\mathrm{d}x\mathrm{d}y\mathrm{d}z\,\overline{V}(x-y,x-z)\varphi_{t}(x)\varphi_{t}(y)\left\langle (\mathcal{N}+1)^{-1/2}\xi,(\mathcal{N}+1)^{(j+1)/2}a_{z}^{*}a_{z}a_{y}a_{x}\psi\right\rangle \\
	& \leq\left(\int\mathrm{d}x\mathrm{d}y\mathrm{d}z\,|\overline{V}(x-y,x-z)|^{2}|\varphi_{t}(x)|^{2}|\varphi_{t}(y)|^{2}\left\Vert a_{z}(\mathcal{N}+1)^{-1/2}\xi\right\Vert ^{2}\right)^{1/2}\\
	&\qquad\times\left(\int\mathrm{d}x\mathrm{d}y\mathrm{d}z\,\left\Vert a_{z}a_{y}a_{x}(\mathcal{N}+1)^{(j+1)/2}\psi\right\Vert ^{2}\right)^{1/2}\\
	& \leq C\|\varphi_{t}\|_{H^{1}(\mathbb{R}^{3})}^{2}\|\xi\|\left\Vert (\mathcal{N}+1)^{(j+4)/2}\psi\right\Vert .
	\end{align*}
	Similarly for $A_{4}^{*}$, we obtain desired lemma.
\end{proof}
\begin{lem}
	\label{lem:L5}Suppose that the assumptions in Theorem \ref{prop:reged-V}
	hold. Then, for any for any $\psi\in\mathcal{F}$ and $j\in\mathbb{N}$,
	there exist a constant $C\equiv C(j)$ such that
	\[
	\left\Vert (\mathcal{N}+1)^{j/2}\mathcal{L}_{5}\psi\right\Vert \leq\frac{Ce^{Kt}}{\alpha_{N}N\sqrt{N}}\left\Vert (\mathcal{N}+1)^{(j+5)/2}\psi\right\Vert .
	\]
\end{lem}

\begin{proof}
	Let 
	\begin{align*}
	A_{5} & =\int\mathrm{d}x\mathrm{d}y\mathrm{d}z\,\overline{V}(x-y,x-z)\varphi_{t}(z)a_{x}^{*}a_{y}^{*}a_{z}^{*}a_{y}a_{x}.
	\end{align*}
	Then
	\[
	\mathcal{L}_{5}=\frac{1}{2N\sqrt{N}}(A_{5}+A_{5}^{*}).
	\]
	Take any $\xi\in\mathcal{F}$. Then
	\begin{align*}
	& \left\langle \xi,(\mathcal{N}+1)^{j/2}A_{5}^{*}\psi\right\rangle \\
	& =\int\mathrm{d}x\mathrm{d}y\mathrm{d}z\,\overline{V}(x-y,x-z)\varphi_{t}(z)\left\langle (\mathcal{N}+1)^{-1/2}\xi,(\mathcal{N}+1)^{(j+1)/2}a_{x}^{*}a_{y}^{*}a_{z}^{*}a_{y}a_{x}\psi\right\rangle \\
	& \leq\left(\int\mathrm{d}x\mathrm{d}y\mathrm{d}z\,\left\Vert a_{x}a_{y}a_{z}(\mathcal{N}+1)^{-3/2}\xi\right\Vert ^{2}\right)^{1/2}\\
	&\qquad\times\left(\int\mathrm{d}x\mathrm{d}y\mathrm{d}z\,|\overline{V}(x-y,x-z)|^{2}|\varphi_{t}(z)|^{2}\left\Vert a_{y}a_{x}(\mathcal{N}+1)^{(j+3)/2}\psi\right\Vert ^{2}\right)^{1/2}\\
	& \leq C\alpha_{N}^{-1}\left(\int\mathrm{d}x\mathrm{d}y\mathrm{d}z\,\left\Vert a_{x}a_{y}a_{z}(\mathcal{N}+1)^{-3/2}\xi\right\Vert ^{2}\right)^{1/2}\\
	&\qquad\times\left(\int\mathrm{d}x\mathrm{d}y\mathrm{d}z\,\left(|\overline{v}(x-z)|^{2}+|\overline{v}(y-z)|^{2}\right)|\varphi_{t}(z)|^{2}\left\Vert a_{y}a_{x}(\mathcal{N}+1)^{(j+3)/2}\psi\right\Vert ^{2}\right)^{1/2}\\
	& \leq C\alpha_{N}^{-1}\|\varphi_{t}\|_{H^{1}(\mathbb{R}^{3})}\|\xi\|\left\Vert (\mathcal{N}+1)^{(j+5)/2}\psi\right\Vert .
	\end{align*}
\end{proof}
\begin{lem}
	\label{lem:NjUphiUtildeUphitildeU} Suppose that the assumptions in
	Theorem \ref{thm:main} hold. Let $\alpha_{N}=N^{-\eta}$. Then, for
	all $j\in\mathbb{N}$, there exist constants $C\equiv C(j)$ and $K\equiv K(j)$
	such that, for any $f\in L^{2}(\mathbb{R}^{3})$, 
	\[
	\left\Vert \left(\mathcal{N}+1\right)^{j/2}\left(\mathcal{U}^{*}\left(t\right)\phi(f)\mathcal{U}\left(t\right)-\mathcal{\widetilde{U}}^{*}\left(t\right)\phi(f)\mathcal{\widetilde{U}}\left(t\right)\right)\Omega\right\Vert \leq\frac{Ce^{Kt}}{N^{2-\eta}}\|f\|_{L^{2}(\mathbb{R}^{3})}.
	\]
\end{lem}

\begin{proof}
	Let 
	\[
	\mathcal{R}_{1}(f):=\left(\mathcal{U}^{*}(t)-\mathcal{\widetilde{U}}^{*}(t)\right)\phi(f)\mathcal{\widetilde{U}}(t)
	\]
	and 
	\[
	\mathcal{R}_{2}(f):=\mathcal{U}^{*}(t)\phi(f)\left(\mathcal{U}(t)-\mathcal{\widetilde{U}}(t)\right)
	\]
	so that 
	\begin{equation}
	\mathcal{U}^{*}(t)\phi(f)\mathcal{U}(t)-\mathcal{\widetilde{U}}^{*}(t)\phi(f)\mathcal{\widetilde{U}}(t)=\mathcal{R}_{1}(f)+\mathcal{R}_{2}(f).\label{eq:R_1}
	\end{equation}
	We begin by estimating the first term in the right-hand side of \eqref{eq:R_1}.
	From Lemma \ref{lem:NjU}, 
	\begin{align*}
	&\left\Vert (\mathcal{N}+1)^{j/2}\mathcal{R}_{1}(f)\Omega\right\Vert \\ & =\left\Vert \int_{0}^{t}\mathrm{d}s(\mathcal{N}+1)^{j/2}\mathcal{U}^{*}(s;0)(\mathcal{L}_{3}+\mathcal{L}_{4}^{c}+\mathcal{L}_{5})\mathcal{\widetilde{U}}^{*}(t)\phi(f)\mathcal{\widetilde{U}}(t)\Omega\right\Vert \\
	& \leq\int_{0}^{t}\mathrm{d}s\left\Vert (\mathcal{N}+1)^{j/2}\mathcal{U}^{*}(s;0)(\mathcal{L}_{3}+\mathcal{L}_{4}^{c}+\mathcal{L}_{5})\mathcal{\widetilde{U}}^{*}(t)\phi(f)\mathcal{\widetilde{U}}(t)\Omega\right\Vert \\
	& \leq Ce^{Kt}\int_{0}^{t}\mathrm{d}s\left\Vert (\mathcal{N}+1)^{j+1}(\mathcal{L}_{3}+\mathcal{L}_{4}^{c}+\mathcal{L}_{5})\mathcal{\widetilde{U}}^{*}(t)\phi(f)\mathcal{\widetilde{U}}(t)\Omega\right\Vert .
	\end{align*}
	From Lemmata \ref{lem:L3}, \ref{lem:L4c}, and \ref{lem:L5}, we
	have
	\begin{align*}
	& \int_{0}^{t}\mathrm{d}s\left\Vert (\mathcal{N}+1)^{j+1}\mathcal{L}_{3}\mathcal{\widetilde{U}}^{*}(t)\phi(f)\mathcal{\widetilde{U}}(t)\Omega\right\Vert \leq\frac{Ce^{Kt}}{\sqrt{N}}\int_{0}^{t}\mathrm{d}s\,\left\Vert (\mathcal{N}+1)^{j+(5/2)}\mathcal{\widetilde{U}}^{*}(t)\phi(f)\mathcal{\widetilde{U}}(t)\Omega\right\Vert \\
	& \leq\frac{Ce^{Kt}}{\sqrt{N}}\int_{0}^{t}\mathrm{d}s\,\left\Vert (\mathcal{N}+1)^{j+4}\phi(f)\mathcal{\widetilde{U}}(t)\Omega\right\Vert \leq\frac{C\|f\|_{L^{2}(\mathbb{R}^{3})}e^{Kt}}{\sqrt{N}}\left(\int_{0}^{t}\mathrm{d}s\,\left\Vert (\mathcal{N}+1)^{j+5}\mathcal{\widetilde{U}}(t)\Omega\right\Vert ^{2}\right)^{1/2}\\
	& \leq\frac{C\|f\|_{L^{2}(\mathbb{R}^{3})}e^{Kt}}{\sqrt{N}}\left(\int_{0}^{t}\mathrm{d}s\,\left\Vert (\mathcal{N}+1)^{j+13/2}\Omega\right\Vert ^{2}\right)^{1/2}\leq\frac{Ce^{Kt}}{\sqrt{N}}\|f\|_{L^{2}(\mathbb{R}^{3})},
	\end{align*}
	\begin{align*}
	& \int_{0}^{t}\mathrm{d}s\left\Vert (\mathcal{N}+1)^{j+1}\mathcal{L}_{4}\mathcal{\widetilde{U}}^{*}(t)\phi(f)\mathcal{\widetilde{U}}(t)\Omega\right\Vert \leq\frac{Ce^{Kt}}{N}\int_{0}^{t}\mathrm{d}s\,\left\Vert (\mathcal{N}+1)^{j+3}\mathcal{\widetilde{U}}^{*}(t)\phi(f)\mathcal{\widetilde{U}}(t)\Omega\right\Vert \\
	& \leq\frac{Ce^{Kt}}{N}\int_{0}^{t}\mathrm{d}s\,\left\Vert (\mathcal{N}+1)^{j+9/2}\phi(f)\mathcal{\widetilde{U}}(t)\Omega\right\Vert \leq\frac{C\|f\|_{L^{2}(\mathbb{R}^{3})}e^{Kt}}{N}\left(\int_{0}^{t}\mathrm{d}s\,\left\Vert (\mathcal{N}+1)^{j+(11/2)}\mathcal{\widetilde{U}}(t)\Omega\right\Vert ^{2}\right)^{1/2}\\
	& \leq\frac{C\|f\|_{L^{2}(\mathbb{R}^{3})}e^{Kt}}{\sqrt{N}}\left(\int_{0}^{t}\mathrm{d}s\,\left\Vert (\mathcal{N}+1)^{j+7}\mathcal{\widetilde{U}}(t)\Omega\right\Vert ^{2}\right)^{1/2}\leq\frac{Ce^{Kt}}{N}\|f\|_{L^{2}(\mathbb{R}^{3})}
	\end{align*}
	and
	\begin{align*}
	& \int_{0}^{t}\mathrm{d}s\left\Vert (\mathcal{N}+1)^{j+1}\mathcal{L}_{5}\mathcal{\widetilde{U}}^{*}(t)\phi(f)\mathcal{\widetilde{U}}(t)\Omega\right\Vert \leq\frac{Ce^{Kt}}{\alpha_{N}N\sqrt{N}}\int_{0}^{t}\mathrm{d}s\,\left\Vert (\mathcal{N}+1)^{j+(7/2)}\mathcal{\widetilde{U}}^{*}(t)\phi(f)\mathcal{\widetilde{U}}(t)\Omega\right\Vert \\
	& \leq\frac{Ce^{Kt}}{\alpha_{N}N\sqrt{N}}\int_{0}^{t}\mathrm{d}s\,\left\Vert (\mathcal{N}+1)^{j+5}\phi(f)\mathcal{\widetilde{U}}(t)\Omega\right\Vert \leq\frac{C\|f\|_{L^{2}(\mathbb{R}^{3})}e^{Kt}}{\alpha_{N}N\sqrt{N}}\left(\int_{0}^{t}\mathrm{d}s\,\left\Vert (\mathcal{N}+1)^{j+6}\mathcal{\widetilde{U}}(t)\Omega\right\Vert ^{2}\right)^{1/2}\\
	& \leq\frac{C\|f\|_{L^{2}(\mathbb{R}^{3})}e^{Kt}}{\sqrt{N}}\left(\int_{0}^{t}\mathrm{d}s\,\left\Vert (\mathcal{N}+1)^{j+15/2}\mathcal{\widetilde{U}}(t)\Omega\right\Vert ^{2}\right)^{1/2}\leq\frac{Ce^{Kt}}{\alpha_{N}N\sqrt{N}}\|f\|_{L^{2}(\mathbb{R}^{3})}.
	\end{align*}
	Then, since the integrand in the right-hand side does not depend on
	$s$, we get 
	\begin{align*}
	\left\Vert (\mathcal{N}+1)^{j/2}\mathcal{R}_{1}(f)\Omega\right\Vert  & \leq Ce^{Kt}\left(\frac{1}{\sqrt{N}}+\frac{1}{N}+\frac{1}{\alpha_{N}N\sqrt{N}}\right)\left\Vert f\right\Vert _{L^{2}(\mathbb{R}^{3})}.
	\end{align*}
	Thus, from Lemma \ref{lem:tildeNj}, we obtain for $\mathcal{R}_{1}(f)$
	that 
	\[
	\left\Vert (\mathcal{N}+1)^{j/2}\mathcal{R}_{1}(f)\Omega\right\Vert \leq C\|f\|e^{Kt}\frac{1}{N^{(3/2)-\eta}}.
	\]
	The study of $\mathcal{R}_{2}(f)$ is similar and we can again obtain
	that 
	\[
	\|(\mathcal{N}+1)^{j/2}\mathcal{R}_{2}(f)\Omega\|\leq C\|f\|e^{Kt}\frac{1}{N^{(3/2)-\eta}}.
	\]
	This completes the proof of the desired lemma.
\end{proof}

\section{Proof of Propositions \ref{prop:Et1} and \ref{prop:Et2}}

\label{sec:Pf-of-Props}

In this section, we prove Propositions \ref{prop:Et1} and \ref{prop:Et2}
by applying the lemmas proved in Section \ref{sec:comparison}. Even
though the proofs are almost the same as previous works, we include
the following proves since underlying lemmas and logic are different.
The structure of the proof is given in \cite{Lee2013}. We, however,
provide this section since the exponents of $(\mathcal{N}+1)$ are
a bit different.
\begin{proof}[Proof of Proposition \ref{prop:Et1}]
	Recall that Comparison dynamics
	\[
	E_{t}^{1}(J)=\frac{d_{N}}{N}\left\langle W^{*}(\sqrt{N}\varphi)\frac{(a^{*}(\varphi))^{N}}{\sqrt{N!}}\Omega,\mathcal{U}^{*}(t)d\Gamma(J)\mathcal{U}(t)\Omega\right\rangle 
	\]
	We begin by 
	\begin{align}
	\left|E_{t}^{1}(J)\right| & =\left|\frac{d_{N}}{N}\left\langle W^{*}(\sqrt{N}\varphi)\frac{(a^{*}(\varphi))^{N}}{\sqrt{N!}}\Omega,\mathcal{U}^{*}(t)d\Gamma(J)\mathcal{U}(t)\Omega\right\rangle \right|\label{eq:E_t^1 1}\\
	& \leq\frac{d_{N}}{N}\left\Vert (\mathcal{N}+1)^{-\frac{1}{2}}W^{*}(\sqrt{N}\varphi)\frac{(a^{*}(\varphi))^{N}}{\sqrt{N!}}\Omega\right\Vert \left\Vert (\mathcal{N}+1)^{\frac{1}{2}}\mathcal{U}^{*}(t)d\Gamma(J)\mathcal{U}(t)\Omega\right\Vert \nonumber 
	\end{align}
	From Lemma \ref{lem:coherent_all}, 
	\begin{equation}
	\left\Vert (\mathcal{N}+1)^{-\frac{1}{2}}W^{*}(\sqrt{N}\varphi)\frac{(a^{*}(\varphi))^{N}}{\sqrt{N!}}\Omega\right\Vert \leq\frac{C}{d_{N}}.\label{eq:E_t^1 2}
	\end{equation}
	By successively applying Lemma \ref{lem:NjU} (and also using the
	inequality \eqref{eq:J-bd}), we also get 
	\begin{align}
	\left\Vert (\mathcal{N}+1)^{\frac{1}{2}}\mathcal{U}^{*}(t)d\Gamma(J)\mathcal{U}(t)\Omega\right\Vert  & \leq Ce^{Kt}\left\Vert (\mathcal{N}+1)^{2}d\Gamma(J)\mathcal{U}(t)\Omega\right\Vert \leq C\left\Vert J\right\Vert e^{Kt}\left\Vert (\mathcal{N}+1)^{3}\mathcal{U}(t)\Omega\right\Vert \nonumber \\
	& \leq C\left\Vert J\right\Vert e^{Kt}\left\Vert (\mathcal{N}+1)^{9/2}\Omega\right\Vert .\label{eq:E_t^1 3}
	\end{align}
	Thus, from \eqref{eq:E_t^1 1}, \eqref{eq:E_t^1 2}, and \eqref{eq:E_t^1 3},
	\[
	\left|E_{t}^{1}(J)\right|\leq\frac{C\|J\|e^{Kt}}{N^{2-\eta}},
	\]
	which proves the desired result. 
\end{proof}
For the proof of Proposition \ref{prop:Et2}, we take almost verbatim
copy of the proof of Lemma 4.2 in \cite{Lee2013}. To make the paper
self-contained, we write it in detail below.
\begin{proof}[Proof of Proposition \ref{prop:Et2}]
	
	Recall the definitions of $\mathcal{R}_{1}$ and $\mathcal{R}_{2}$
	in the pr2oof of Lemma \ref{lem:NjUphiUtildeUphitildeU}. Let $\mathcal{R}=\mathcal{R}_{1}+\mathcal{R}_{2}$
	so that 
	\[
	\mathcal{R}(f)=\mathcal{U}^{*}(t)\phi(f)\mathcal{U}(t)-\mathcal{\widetilde{U}}^{*}(t)\phi(f)\mathcal{\widetilde{U}}(t).
	\]
	From the parity conservation \eqref{eq:Parity_Consevation}, 
	\[
	P_{2k}\mathcal{\widetilde{U}}^{*}(t)\phi(J\varphi_{t})\mathcal{\widetilde{U}}(t)\Omega=0
	\]
	for all $k=0,1,\dots$. (See Lemma 8.2 in \cite{Lee2013} for more
	detail.) Thus, 
	\begin{align}
	&\left|E_{t}^{2}(J)\right| \\& =\frac{d_{N}}{\sqrt{N}}\left\langle \frac{(a^{*}(\varphi))^{N}}{\sqrt{N!}}\Omega,W^{*}(\sqrt{N}\varphi)\mathcal{\widetilde{U}}^{*}(t)\phi(J\varphi_{t})\mathcal{\mathcal{\widetilde{U}}}(t)\Omega\right\rangle \nonumber \\
	& \qquad+\frac{d_{N}}{\sqrt{N}}\left\langle \frac{(a^{*}(\varphi))^{N}}{\sqrt{N!}}\Omega,W^{*}(\sqrt{N}\varphi)\mathcal{R}(J\varphi_{t})\Omega\right\rangle \nonumber \\
	& \leq\frac{d_{N}}{\sqrt{N}}\left\Vert \sum_{k=1}^{\infty}(\mathcal{N}+1)^{-\frac{5}{2}}P_{2k-1}W^{*}(\sqrt{N}\varphi)\frac{(a^{*}(\varphi))^{N}}{\sqrt{N!}}\Omega\right\Vert \left\Vert (\mathcal{N}+1)^{\frac{5}{2}}\mathcal{\widetilde{U}}^{*}(t)\phi(J\varphi_{t})\mathcal{\widetilde{U}}(t)\Omega\right\Vert \nonumber \\
	& \qquad+\frac{d_{N}}{\sqrt{N}}\left\Vert (\mathcal{N}+1)^{-\frac{1}{2}}W^{*}(\sqrt{N}\varphi)\frac{(a^{*}(\varphi))^{N}}{\sqrt{N!}}\Omega\right\Vert \left\Vert (\mathcal{N}+1)^{\frac{1}{2}}\mathcal{R}(J\varphi_{t})\Omega\right\Vert \label{eq:e_t^2 1}
	\end{align}
	Let $K=\frac{1}{2}N^{1/3}$ so that Lemmas \ref{lem:coherent_all}
	and \ref{lem:coherent_even_odd} show that 
	\begin{align*}
	& \left\Vert \sum_{k=1}^{\infty}(\mathcal{N}+1)^{-\frac{5}{2}}P_{2k-1}W^{*}(\sqrt{N}\varphi)\frac{(a^{*}(\varphi))^{N}}{\sqrt{N!}}\Omega\right\Vert ^{2}\\
	& \qquad\leq\sum_{k=1}^{K}\left\Vert (\mathcal{N}+1)^{-\frac{5}{2}}P_{2k-1}W^{*}(\sqrt{N}\varphi)\frac{(a^{*}(\varphi))^{N}}{\sqrt{N!}}\Omega\right\Vert ^{2}\\
	& \qquad\qquad+\frac{1}{K^{4}}\sum_{k=K}^{\infty}\left\Vert (\mathcal{N}+1)^{-1/2}P_{2k-1}W^{*}(\sqrt{N}\varphi)\frac{(a^{*}(\varphi))^{N}}{\sqrt{N!}}\Omega\right\Vert ^{2}\\
	& \qquad\leq\left(\sum_{k=1}^{K}\frac{C}{k^{2}d_{N}^{2}N}\right)+\frac{C}{N^{4/3}}\left\Vert (\mathcal{N}+1)^{-1/2}W^{*}(\sqrt{N}\varphi)\frac{(a^{*}(\varphi))^{N}}{\sqrt{N!}}\Omega\right\Vert \leq\frac{C}{d_{N}^{2}N}.
	\end{align*}
	Using Lemma \ref{lem:tildeNj}, 
	\begin{alignat*}{1}
	& \left\Vert (\mathcal{N}+1)^{\frac{5}{2}}\mathcal{\widetilde{U}}^{*}(t)\phi(J\varphi_{t})\mathcal{\widetilde{U}}(t)\Omega\right\Vert \leq Ce^{Kt}\left\Vert (\mathcal{N}+1)^{4}\phi(J\varphi_{t})\mathcal{\widetilde{U}}(t)\Omega\right\Vert \\
	& \quad\leq C\|J\varphi_{t}\|e^{Kt}\left\Vert (\mathcal{N}+1)^{5}\mathcal{\mathcal{\widetilde{U}}}(t)\Omega\right\Vert \leq C\|J\|e^{Kt}\left\Vert (\mathcal{N}+1)^{13/2}\Omega\right\Vert \leq C\|J\|e^{Kt}.
	\end{alignat*}
	For the second term in the right-hand side of \eqref{eq:e_t^2 1},
	we use Lemmas \ref{lem:coherent_all} and \ref{lem:NjUphiUtildeUphitildeU},
	where we put $J\varphi_{t}$ in the place of $f$ for the latter.
	Altogether, we have 
	\[
	\left\Vert (\mathcal{N}+1)^{j/2}\mathcal{R}(f)\Omega\right\Vert \leq\frac{C\|f\|e^{Kt}}{N^{2-\eta}},
	\]
	which is the desired conclusion.
\end{proof}

\newpage
\appendix

\section{Standard Fock space formalism\label{sec:Fock_space}}

This section is devoted to introduce the standard Fock space formalism.
One can see more details in many articles, for example, \cite{Benedikter2016,ChenLeeSchlein2011,Rodnianski2009}.

To consider the system of $N$-bosons, we want to embed our the system
into a larger space so-called bosonic Fock space over $L^{2}(\mathbb{R}^{3})$
which is defined as
\[
\mathcal{F}:=\bigoplus_{n\geq0}L^{2}(\mathbb{R}^{3},\mathrm{d}x)^{\otimes_{s}n}=\mathbb{C}\oplus\bigoplus_{n\geq1}L_{s}^{2}(\mathbb{R}^{3n},\mathrm{d}x_{1},\dots,\mathrm{d}x_{n}),
\]
where $L_{s}^{2}=L_{s}^{2}(\mathbb{R}^{3n},\mathrm{d}x_{1},\dots,\mathrm{d}x_{n})$
is a symmetric subspace of $L^{2}(\mathbb{R}^{3n},\mathrm{d}x_{1},\dots,\mathrm{d}x_{n})$
where we let $L_{s}^{2}(\mathbb{R}^{3})^{\otimes0}=\mathbb{C}$. An
element $\psi\in\mathcal{F}$ is called a state, and it can be understood
as a sequence $\psi=\{\psi^{\left(n\right)}\}_{n\geq0}$ of $n$-particle
wave functions $\psi^{\left(n\right)}\in L_{s}^{2}(\mathbb{R}^{3n})$.
The inner product on $\mathcal{F}$ is defined by 
\[
\begin{aligned}\langle\psi_{1},\psi_{2}\rangle & =\sum_{n\geq0}\langle\psi_{1}^{\left(n\right)},\psi_{2}^{\left(n\right)}\rangle_{L^{2}(\mathbb{R}^{3n})}\\
& =\overline{\psi_{1}^{(0)}}\psi_{2}^{(0)}+\sum_{n\geq0}\int\mathrm{d}x_{1}\dots\mathrm{d}x_{n}\overline{\psi_{1}^{\left(n\right)}(x_{1},\dots,x_{n})}\psi_{2}^{\left(n\right)}(x_{1},\dots,x_{n}).
\end{aligned}
\]
A vacuum state is defined as $\Omega:=\left\{ 1,0,0,\dots\right\} \in\mathcal{F}$.
Since a state in a Fock space can have a different number of particles,
we define the creation operator $a^{*}(f)$ and the annihilation operator
$a(f)$ on $\mathcal{F}$ by 
\begin{equation}
\left(a^{*}\left(f\right)\psi\right)^{\left(n\right)}(x_{1},\dots,x_{n})=\frac{1}{\sqrt{n}}\sum_{j=1}^{n}f(x_{j})\psi^{\left(n-1\right)}(x_{1},\dots,x_{j-1},x_{j+1},\dots,x_{n})\label{eq:creation}
\end{equation}
and 
\begin{equation}
\left(a\left(f\right)\psi\right)^{\left(n\right)}(x_{1},\dots,x_{n})=\sqrt{n+1}\int\mathrm{d}x\overline{f\left(x\right)}\psi^{\left(n+1\right)}(x,x_{1},\dots,x_{n})\label{eq:annihilation}
\end{equation}
which creates a particle $f$ to the system and annihilates $f$ from
the system (respectively). Note that both $a^{*}(f)$ and $a(f)$
are not self-adjoint. We define the self-adjoint operator $\phi(f)$
such as 
\[
\phi(f)=a^{*}(f)+a(f).
\]
We also use operator-valued distributions $a_{x}^{*}$ and $a_{x}$
satisfying 
\[
a^{*}(f)=\int\mathrm{d}x\,f\left(x\right)a_{x}^{*},\qquad a(f)=\int\mathrm{d}x\,\overline{f\left(x\right)}a_{x}
\]
for any $f\in L^{2}(\mathbb{R}^{3})$. The canonical commutation relation
between creation and annihilation operators is 
\[
[a(f),a^{*}(g)]=\langle f,g\rangle_{L^{2}(\mathbb{R}^{3})},\quad[a(f),a(g)]=[a^{*}(f),a^{*}(g)]=0,
\]
which also assumes the form 
\[
[a_{x},a_{y}^{*}]=\delta\left(x-y\right),\quad[a_{x},a_{y}]=[a_{x}^{*},a_{y}^{*}]=0.
\]
Moreover, the number operator $\mathcal{N}$, which gives us the expected
number of the state in Fock space, is defined by 
\begin{equation}
\mathcal{N}:=\int\mathrm{d}x\,a_{x}^{*}a_{x},\label{eq:number operator}
\end{equation}
For each non-negative integer $n$, we introduce the projection operator
onto the $n$-particle sector of the Fock space, 
\[
P_{n}(\psi):=(0,0,\dots,0,\psi^{(n)},0,\dots)
\]
for $\psi=(\psi^{(0)},\psi^{(1)},\dots)\in\mathcal{F}$. For simplicity,
with slight abuse of notation, we will use $\psi^{(n)}$ to denote
$P_{n}\psi$. and it satisfies that $(\mathcal{N}\psi)^{(n)}=n\psi^{(n)}$.
For an operator $J$ on the one-particle sector $L^{2}(\mathbb{R}^{3},\mathrm{d}x)$,
we define its second quantization $d\Gamma(J)$ by
\[
\left(d\Gamma\left(J\right)\psi\right)^{(n)}=\sum_{j=1}^{n}J_{j}\psi^{(n)}
\]
where $J_{j}=1\otimes\dots\otimes J\otimes\dots\otimes1$ is the operator
$J$ acting on the $j$-th variable only. The number operator defined
above can also be understood as the second quantization of the identity,
i.e., $\mathcal{N}=d\Gamma(1)$. With a kernel $J(x;y)$ of the operator
$J$, the second quantization $d\Gamma(J)$ can be also be written
as 
\[
d\Gamma(J)=\int\mathrm{d}x\mathrm{d}y\,J(x;y)a_{x}^{*}a_{y},
\]
which is consistent with \eqref{eq:number operator}.

Since the annihilation operator and the creation operator forms the
number operator, it is natural to control the operators by the number
operator. To control the operators and second quantization, we provide
the following lemma.
\begin{lem}
	For $\alpha>0$, let $D(\mathcal{N}^{\alpha})=\{\psi\in\mathcal{F}:\sum_{n\geq1}n^{2\alpha}\|\psi^{(n)}\|^{2}<\infty\}$
	denote the domain of the operator $\mathcal{N}^{\alpha}$. For any
	$f\in L^{2}(\mathbb{R}^{3},dx)$ and any $\psi\in D(\mathcal{N}^{1/2})$,
	we have 
	\begin{equation}
	\begin{split}\|a(f)\psi\| & \leq\|f\|\,\|\mathcal{N}^{1/2}\psi\|,\\
	\|a^{*}(f)\psi\| & \leq\|f\|\,\|(\mathcal{N}+1)^{1/2}\psi\|,\\
	\|\phi(f)\psi\| & \leq2\|f\|\|\left(\mathcal{N}+1\right)^{1/2}\psi\|\,.
	\end{split}
	\label{eq:bd-a}
	\end{equation}
	Moreover, for any bounded one-particle operator $J$ on $L^{2}(\mathbb{R}^{3},dx)$
	and for every $\psi\in D(\mathcal{N})$, we find 
	\begin{equation}
	\|d\Gamma(J)\psi\|\leq\|J\|\|\mathcal{N}\psi\|\,.\label{eq:J-bd}
	\end{equation}
\end{lem}

\begin{proof}
	See \cite[Lemma 2.1]{Rodnianski2009} for \eqref{eq:bd-a}, and see
	\cite[Lemma 3.1]{ChenLeeSchlein2011} for \eqref{eq:J-bd}.
\end{proof}
Heuristically, there are eigenvectors of $a_{x}$ with the eigenvalue
$\sqrt{N}f$, where $f\in L^{2}(\mathbb{R}^{3})$. It known as the
coherent states, defined by, for $f\in L^{2}(\mathbb{R}^{3})$,
\[
\psi_{\mathrm{coh}}(f):=e^{-\|f\|_{L^{2}}^{2}/2}\sum_{n\geq0}\frac{\left(a^{*}(f)\right)^{n}}{n!}\Omega=e^{-\|f\|_{L^{2}}^{2}/2}\sum_{n\geq0}\frac{1}{\sqrt{n!}}f^{\otimes n}.
\]
Then from \eqref{eq:Kernel_gamma} one obtain $\gamma_{\psi_{\mathrm{coh}}}^{(1)}(x;y)=\varphi_{t}(x)\overline{\varphi_{t}(y)}$,
which is exactly the one-particle marginal density associated with
the factorized wave function $\varphi_{t}^{\otimes N}$. Note that,
unlike our system with $N$-particles, such eigenvectors of the annihilation
operator can have any number of particles. We, however, can utilize
coherent states for our goal.

The coherent state can be generated by acting Weyl operator $W(f)$
on vacuum state $\Omega$. i.e.,
\begin{equation}
\psi_{\mathrm{coh}}(f)=W(f)\Omega=e^{-\left\Vert f\right\Vert ^{2}/2}\exp\left(a^{*}\left(f\right)\right)\Omega=e^{-\left\Vert f\right\Vert ^{2}/2}\sum_{n\geq0}\frac{1}{\sqrt{n!}}f^{\otimes n}.\label{Weyl_f}
\end{equation}
Where the Weyl operator $W\left(f\right)$ is defined by 
\[
W(f):=\exp\left(a^{*}(f)-a(f)\right)
\]
and it also satisfies 
\[
W\left(f\right)=e^{-\left\Vert f\right\Vert ^{2}/2}\exp\left(a^{*}(f)\right)\exp\left(-a(f)\right),
\]
which is known as the Hadamard lemma in Lie algebra. We collect the
useful properties of the Weyl operator and the coherent states in
the following lemma.
\begin{lem}
	\label{lem:Basic_Weyl} Let $f,g\in L^{2}(\mathbb{R}^{3})$.
	\begin{enumerate}
		\item The commutation relation between the Weyl operators is given by 
		\[
		W(f)W(g)=W(g)W(f)e^{-2\mathrm{i}\cdot\mathrm{Im}\langle f,g\rangle}=W(f+g)e^{-\mathrm{i}\cdot\mathrm{Im}\langle f,g\rangle}.
		\]
		\item The Weyl operator is unitary and satisfies that 
		\[
		W(f)^{*}=W(f)^{-1}=W(-f).
		\]
		\item The coherent states are eigenvectors of annihilation operators, i.e.,
		\[
		a_{x}\psi(f)=f(x)\psi(f)\quad\Rightarrow\quad a(g)\psi(f)=\langle g,f\rangle_{L^{2}}\psi(f).
		\]
		The commutation relation between the Weyl operator and the annihilation
		operator (or the creation operator) is thus 
		\[
		W^{*}(f)a_{x}W(f)=a_{x}+f(x)\quad\text{and}\quad W^{*}(f)a_{x}^{*}W(f)=a_{x}^{*}+\overline{f(x)}.
		\]
		\item The distribution of $\mathcal{N}$ with respect to the coherent state
		$\psi\left(f\right)$ is Poisson. In particular, 
		\[
		\langle\psi(f),\mathcal{N}\psi(f)\rangle=\|f\|^{2},\qquad\langle\psi(f),\mathcal{N}^{2}\psi(f)\rangle-\langle\psi(f),\mathcal{N}\psi(f)\rangle^{2}=\|f\|^{2}.
		\]
	\end{enumerate}
\end{lem}

We omit the proof of the lemma, since it can be derived from the definition
of the Weyl operator and elementary calculation.

For 
\begin{equation}
d_{N}:=\frac{\sqrt{N!}}{N^{N/2}e^{-N/2}},\label{eq:d_N}
\end{equation}
we note that $C^{-1}N^{1/4}\leq d_{N}\leq CN^{1/4}$ for some constant
$C>0$ independent of $N$, which can be easily checked by using Stirling's
formula. Then we have the following lemmas.
\begin{lem}
	\label{lem:coherent_all} There exists a constant $C>0$ independent
	of $N$ such that, for any $f\in L^{2}(\mathbb{R}^{3})$ with $\|f\|_{L^{2}(\mathbb{R}^{3})}=1$,
	we have 
	\[
	\left\Vert (\mathcal{N}+1)^{-1/2}W^{*}(\sqrt{N}f)\frac{(a^{*}(f))^{N}}{\sqrt{N!}}\Omega\right\Vert \leq\frac{C(t)}{d_{N}}.
	\]
\end{lem}

\begin{proof}
	See \cite[Lemma 6.3]{ChenLee2011}. 
\end{proof}
\begin{lem}
	\label{lem:coherent_even_odd} Let $P_{m}$ be the projection onto
	the $m$-particle sector of the Fock space $\mathcal{F}$ for a non-negative
	integer $m$. Then, for any non-negative integers $k\leq(1/2)N^{1/3}$
	and for any $f\in L^{2}(\mathbb{R}^{3})$ with $\|f\|_{L^{2}(\mathbb{R}^{3})}=1$,
	\[
	\left\Vert P_{2k}W^{*}(\sqrt{N}f)\frac{(a^{*}(f))^{N}}{\sqrt{N!}}\Omega\right\Vert \leq\frac{2}{d_{N}}
	\]
	and 
	\[
	\left\Vert P_{2k+1}W^{*}(\sqrt{N}f)\frac{(a^{*}(f))^{N}}{\sqrt{N!}}\Omega\right\Vert \leq\frac{2(k+1)^{3/2}}{d_{N}\sqrt{N}}.
	\]
\end{lem}

\begin{proof}
	See \cite[Lemma 7.2]{Lee2013}.
\end{proof}

\vspace{2em}

\section{Properties of the solution of quintic Hartree equation\label{sec:properties_quntic_Hartree}}

In this section, our goal is to bound
\begin{equation}
\int\mathrm{d}x\mathrm{d}y\mathrm{d}z\,|V(x-y,x-z)|^{2}|\varphi_{t}(x)|^{2}|\varphi_{t}(y)|^{2}|\varphi_{t}(z)|^{2}\label{eq:V2ph6}
\end{equation}
which will appear in the proves given in Section \ref{sec:comparison}.
Note that it is different from the potential energy because we have
a square for $V$.

The following lemma cannot be directly applied for our purpose. We
offer it, however, to provide an intuition to the readers.
\begin{lem}[generalized Young's inequality]
	\label{lem:gYoung}Let $p_{j}>1$ for each $j=1,\dots,5$ with
	\[
	\sum_{j=1}^{5}\frac{1}{p_{j}}=3.
	\]
	Let $f_{j}\in L^{p_{j}}(\mathbb{R}^{n})$ for each $j=1,\dots,5$.
	Then there exists a constant $C(n,\{p_{j}\}_{j=1}^{5})$, independent
	of $f_{j}$, such that
	\begin{align*}
	&\left|\int_{\mathbb{R}^{n}}\int_{\mathbb{R}^{n}}\int_{\mathbb{R}^{n}}\mathrm{d}x\mathrm{d}y\mathrm{d}z\,f_{1}(x)f_{2}(y)f_{3}(z)f_{4}(x-y)f_{3}(x-z)\right|\\
	&\leq C(n,\{p_{j}\}_{j=1}^{5})\prod_{j=1}^{5}\|f_{j}\|_{p_{j}}.\label{eq:gYineq}
	\end{align*}
\end{lem}

\begin{proof}
	Let 
	\[
	I=\int_{\mathbb{R}^{n}}\int_{\mathbb{R}^{n}}\int_{\mathbb{R}^{n}}\mathrm{d}x\mathrm{d}y\mathrm{d}z\,f_{1}(x)f_{2}(y)f_{3}(z)f_{4}(x-y)f_{3}(x-z)
	\]
	Then, we integrate $y$ and $z$ first so that we have
	\[
	I=\int_{\mathbb{R}^{n}}\mathrm{d}x\,f_{1}(x)\left(f_{2}*f_{4}\right)(x)\left(f_{3}*f_{5}\right)(x).
	\]
	Using Hölder inequality,
	\begin{equation}
	|I|\leq\|f_{1}\|_{p_{1}}\|f_{2}*f_{4}\|_{q}\|f_{3}*f_{5}\|_{r}\label{eq:Holder-for-gYineq}
	\end{equation}
	where
	\begin{equation}
	\frac{1}{p_{1}}+\frac{1}{q}+\frac{1}{r}=1.\label{eq:Holder-expo-for-gYineq}
	\end{equation}
	By Young's convolutional inequality,
	\begin{equation}
	\|f_{2}*f_{4}\|_{q}\leq\|f_{2}\|_{p_{2}}\|f_{4}\|_{p_{4}}\qquad\text{and}\qquad\|f_{3}*f_{5}\|_{r}\leq\|f_{3}\|_{p_{3}}\|f_{5}\|_{p_{5}}\label{eq:Y-for-gYineq}
	\end{equation}
	with 
	\begin{equation}
	\frac{1}{q}+1=\frac{1}{p_{2}}+\frac{1}{p_{4}}\qquad\text{and}\qquad\frac{1}{r}+1=\frac{1}{p_{3}}+\frac{1}{p_{5}}.\label{eq:Yconv-expo-for-gYineq}
	\end{equation}
	Then combining \eqref{eq:Holder-for-gYineq} and \eqref{eq:Y-for-gYineq}
	together with \eqref{eq:Holder-expo-for-gYineq} and \eqref{eq:Yconv-expo-for-gYineq},
	we 
	\[
	\left|I\right|\leq C(n,\{p_{j}\}_{j=1}^{5})\prod_{j=1}^{5}\|f_{j}\|_{p_{j}}
	\]
	with
	\[
	\sum_{j=1}^{5}\frac{1}{p_{j}}=3.
	\]
\end{proof}
If we allow $f_{4}\in L_{w}^{p_{4}}$ and $f_{5}\in L_{w}^{p_{5}}$
instead of $L^{p_{j}}$, we can utilize such lemma for our Coulomb
singularity. According to Lieb and Loss in \cite{lieb2001analysis},
Hardy-Littlewood-Sobolev inequality can be understood as a weak Young's
inequality. Hence, by we provide the following lemma, which generalize
Hardy-Littlewood-Sobolev inequality.
\begin{lem}
	\label{lem:gHLSineq}Let $p_{1},p_{2},p_{3}>1$ and $0<\lambda_{1},\lambda_{2}<n$
	with
	\[
	\frac{1}{p_{1}}+\frac{1}{p_{2}}+\frac{1}{p_{3}}+\frac{\lambda_{1}+\lambda_{2}}{n}=3.
	\]
	Let $f_{j}\in L^{p_{j}}(\mathbb{R}^{n})$. Then there exists a constant
	$C(n,\lambda,p_{1},p_{2},p_{3})$, independent of $f_{j}$, such that
	\begin{align*}
	&\left|\int_{\mathbb{R}^{n}}\int_{\mathbb{R}^{n}}\int_{\mathbb{R}^{n}}\mathrm{d}x\mathrm{d}y\mathrm{d}z\,\frac{1}{|x-y|^{\lambda_{1}}}\frac{1}{|x-z|^{\lambda_{2}}}f_{1}(x)f_{2}(y)f_{3}(z)\right|\\
	&\leq C(n,\lambda,p_{1},p_{2},p_{3})\prod_{j=1}^{3}\|f_{j}\|_{p_{j}}.\label{eq:gHLSineq}
	\end{align*}
\end{lem}

\begin{proof}
	This proof generalize the proof of Theorem 4.3 in \cite[pp.108-110]{lieb2001analysis}.
	The lemma is followed by applying twice the normal Hardy-Littlewood-Sobolev-inequality
	(in the first step, call $g=|.|^{-\lambda}*f_{2}$).
\end{proof}
Using this lemma, we can prove the boundedness of $H^{1}(\mathbb{R}^{3})$-norm
as follows.
\begin{lem}[Boundedness of $H^{1}(\mathbb{R}^{3})$-norm]
	\label{lem:H1bdd}Let $\varphi_{t}$ be the solution of quintic Hartree
	equation with initial data $\varphi_{0}$. If $\|\varphi_{0}\|_{H^{1}}<C$,
	then
	\[
	\|\varphi_{t}\|_{H^{1}}\leq C.
	\]
\end{lem}

\begin{proof}
	Define the energy $\mathcal{E}(\varphi_{t})$ by
	\begin{align*}
	\mathcal{E}(\varphi_{t}) & :=\frac{1}{2}\int\mathrm{d}x\,|\nabla\varphi_{t}(x)|^{2}+\frac{\lambda}{6}\int\mathrm{d}x\mathrm{d}y\mathrm{d}z\,V(x-y,x-z)\varphi_{t}(x)|^{2}|\varphi_{t}(y)|^{2}|\varphi_{t}(z)|^{2}\\
	& \geq\frac{1}{2}\int\mathrm{d}x\,|\nabla\varphi_{t}(x)|^{2}-\\
	& \qquad+\frac{\lambda}{6}\int\mathrm{d}x\mathrm{d}y\mathrm{d}z\,\Big(v(x-y)v(x-z)+v(y-z)v(y-x)\\
	&\qquad\qquad\qquad+v(z-x)v(z-y)\Big)|\varphi_{t}(x)|^{2}|\varphi_{t}(y)|^{2}|\varphi_{t}(z)|^{2}.
	\end{align*}
	By Lemma \ref{lem:gHLSineq}, we obtain
	\[
	\int\mathrm{d}x\mathrm{d}y\mathrm{d}z\,\frac{1}{|x-y|}\frac{1}{|x-z|}|\varphi_{t}(x)|^{2}|\varphi_{t}(y)|^{2}|\varphi_{t}(z)|^{2}\leq C\||\varphi|^{2}\|_{L^{9/8}}^{3}.
	\]
	By Riesz-Thorin interpolation theorem, one get
	\[
	\||\varphi|^{2}\|_{L^{9/8}}^{3}=\|\varphi\|_{L^{9/4}}^{6}\leq\|\varphi\|_{L^{2}}^{5}\|\varphi\|_{L^{6}}\leq\|\varphi\|_{L^{2}}^{5}\|\varphi\|_{H^{1}}.
	\]
	Hence, for sufficiently small $\varepsilon>0$ so that 
	\[
	\frac{1}{2}\int\mathrm{d}x\,|\nabla\varphi_{t}(x)|^{2}-C\varepsilon\|\varphi_{t}\|_{H^{1}}^{2}>0,
	\]
	we have
	\begin{align*}
	\mathcal{E}(\varphi_{0})=\mathcal{E}(\varphi_{t}) & \geq\frac{1}{2}\int\mathrm{d}x\,|\nabla\varphi_{t}(x)|^{2}-C\|\varphi_{t}\|_{H^{1}}^{1}\\
	& \geq\frac{1}{2}\int\mathrm{d}x\,|\nabla\varphi_{t}(x)|^{2}-C\left(\varepsilon\|\varphi_{t}\|_{H^{1}}^{2}+\frac{1}{\varepsilon}\right).
	\end{align*}
	Thus,
	\[
	\mathcal{E}(\varphi_{0})+\frac{C}{\varepsilon}\geq C\|\varphi_{t}\|_{H^{1}}^{2}.
	\]
	This leads us to the conclusion.
\end{proof}
\begin{rem}
	For $\lambda>-1/4$, the proof of Lemma \ref{lem:H1bdd} is a bit
	easier as follows.
	
	If $\lambda>0$, we are done since
	\[
	\mathcal{E}(\varphi_{t})\geq\frac{1}{2}\int\mathrm{d}x\,|\nabla\varphi_{t}(x)|^{2}
	\]
	implies that $\|\varphi_{t}\|_{H^{1}}\leq C(\mathcal{E}(\varphi_{0})+\|\varphi_{0}\|_{L^{2}})$.
	
	Note that 
	\[
	\int\mathrm{d}x\mathrm{d}y\mathrm{d}z\,\frac{1}{|x-y|}\frac{1}{|x-z|}|\varphi_{t}(x)|^{2}|\varphi_{t}(y)|^{2}|\varphi_{t}(z)|^{2}\leq2\|\varphi_{t}\|_{H^{1}}^{2}\|\varphi_{t}\|_{L^{2}}^{4}
	\]
	implies
	\[
	\mathcal{E}(\varphi_{0})=\mathcal{E}(\varphi_{t})\geq\frac{1}{2}\int\mathrm{d}x\,|\nabla\varphi_{t}(x)|^{2}-2|\lambda|\|\varphi_{t}\|_{H^{1}}^{2}\|\varphi_{t}\|_{L^{2}}^{4}\geq\left(\frac{1}{2}-2|\lambda|\|\varphi_{t}\|_{L^{2}}^{4}\right)\|\varphi_{t}\|_{H^{1}}^{2}
	\]
	Then, for $\lambda$satisfying
	\[
	|\lambda|<\frac{1}{4},
	\]
	we have 
	\[
	\|\varphi_{t}\|_{H^{1}}\leq C.
	\]
\end{rem}

Then we have the following lemma which is our goal of this section.
\begin{lem}
	\label{lem:V2phiphi}Let $\varphi_{t}$ be the solution of quintic
	Hartree equation with three-body interaction potential $V(x-y,x-z)$
	having initial data $\varphi_{0}\in H^{1}$. Then
	\[
	\int\mathrm{d}x\mathrm{d}y\mathrm{d}z\,|V(x-y,x-z)|^{2}|\varphi_{t}(x)|^{2}|\varphi_{t}(y)|^{2}|\varphi_{t}(z)|^{2}\leq C.
	\]
\end{lem}

\begin{proof}
	First, we rewrite
	
	\begin{align*}
	& \int\mathrm{d}x\mathrm{d}y\,|V(x-y,x-z)|^{2}|\varphi_{t}(x)|^{2}|\varphi_{t}(y)|^{2}\\
	& =\lambda^{2}\int\mathrm{d}x\mathrm{d}y\,\left(\frac{1}{|x-y|}\frac{1}{|x-z|}+\frac{1}{|y-x|}\frac{1}{|y-z|}+\frac{1}{|z-x|}\frac{1}{|z-y|}\right)^{2}|\varphi_{t}(x)|^{2}|\varphi_{t}(y)|^{2}\\
	& \leq2\lambda^{2}\int\mathrm{d}x\mathrm{d}y\,\frac{1}{|x-y|^{2}}\frac{1}{|x-z|^{2}}|\varphi_{t}(x)|^{2}|\varphi_{t}(y)|^{2}\\
	& \qquad+2\lambda^{2}\int\mathrm{d}x\mathrm{d}y\,\frac{1}{|x-y|^{2}}\frac{1}{|x-z|^{2}}|\varphi_{t}(x)|^{2}|\varphi_{t}(y)|^{2}\\
	& \qquad+2\lambda^{2}\int\mathrm{d}x\mathrm{d}y\,\frac{1}{|x-y|^{2}}\frac{1}{|x-z|^{2}}|\varphi_{t}(x)|^{2}|\varphi_{t}(y)|^{2}\\
	& =:I_{1}+I_{2}+I_{3}
	\end{align*}
	By Lemma \ref{lem:H1bdd}, we can prove that
	\[
	I_{1}=2\lambda^{2}\int\mathrm{d}x\mathrm{d}y\,\frac{1}{|x-y|^{2}}\frac{1}{|x-z|^{2}}|\varphi_{t}(x)|^{2}|\varphi_{t}(y)|^{2}\leq C\|\varphi_{t}\|_{H^{1}}^{4}\leq C.
	\]
	Similarly for other terms, we get the conclusion.
\end{proof}

\section*{Acknowledgment}

I would like to thank Younghun Hong, Ji Oon Lee, Benjamin Schlein, and Peter Pickl for
helpful discussion. I also acknowledge an anonymous reviewer. Jinyeop
Lee was partially supported by Samsung Science and Technology Foundation
(SSTF-BA1401-51) and by the National Research Foundation of Korea(NRF)
grant funded by the Korea government(MSIT) (NRF-2019R1A5A1028324).

\end{document}